\documentclass[review]{elsarticle}

\usepackage{hyperref}
\usepackage{graphicx}
\usepackage{epsfig}
\usepackage{amssymb}
\usepackage{amsmath}
\usepackage{amsfonts}
\usepackage{array}
\usepackage{booktabs}
\usepackage{algorithm}
\usepackage{algorithmic}
\usepackage{multirow}
\usepackage{multicol}
\usepackage{subfigure}
\usepackage{color}
\usepackage{tcolorbox}
\usepackage{enumerate}
\usepackage{diagbox}
\usepackage{amsthm}

\newtheorem{theorem}{Theorem}[section]
\newtheorem{lemma}[theorem]{Lemma}
\newtheorem{note}[theorem]{NOTE}

\journal{Knowledge-Based Systems}

\begin{document}

\begin{frontmatter}

\title{A Semantic-Rich Similarity in Heterogeneous Information Networks}

\author{Yu Zhou}
\ead{peterjone85@hotmail.com}
\address{School of Software, Xidian University, Xi'an, Shaanxi, China}

\author{Jianbin Huang\footnote{Corresponding Author}}
\ead{jbhuang@xidian.edu.cn}
\address{School of Software, Xidian University, Xi'an, Shaanxi, China}

\author{Heli Sun}
\ead{hlsun@mail.xjtu.edu.cn}
\address{Computer Science and Technology, Xi'an Jiaotong University, Xi'an, Shaanxi, China}
%
%


%
%

\begin{abstract}
  Measuring the similarities between objects in information networks has fundamental importance in recommendation systems, clustering and web search.
  The existing metrics depend on the meta path or meta structure specified by users.
  In this paper, we propose a stratified meta structure based similarity $SMSS$ in heterogeneous information networks.
  The stratified meta structure can be constructed automatically and capture rich semantics.
  Then, we define the commuting matrix of the stratified meta structure by virtue of the commuting matrices of meta paths and meta structures.
  As a result, $SMSS$ is defined by virtue of these commuting matrices.
  Experimental evaluations show that the proposed $SMSS$ on the whole outperforms the state-of-the-art metrics in terms of ranking and clustering.
\end{abstract}


\end{frontmatter}


\section{Introduction}\label{sec:introduction}

\begin{figure}[htb]
  \centering
  \includegraphics[width=0.7\textwidth]{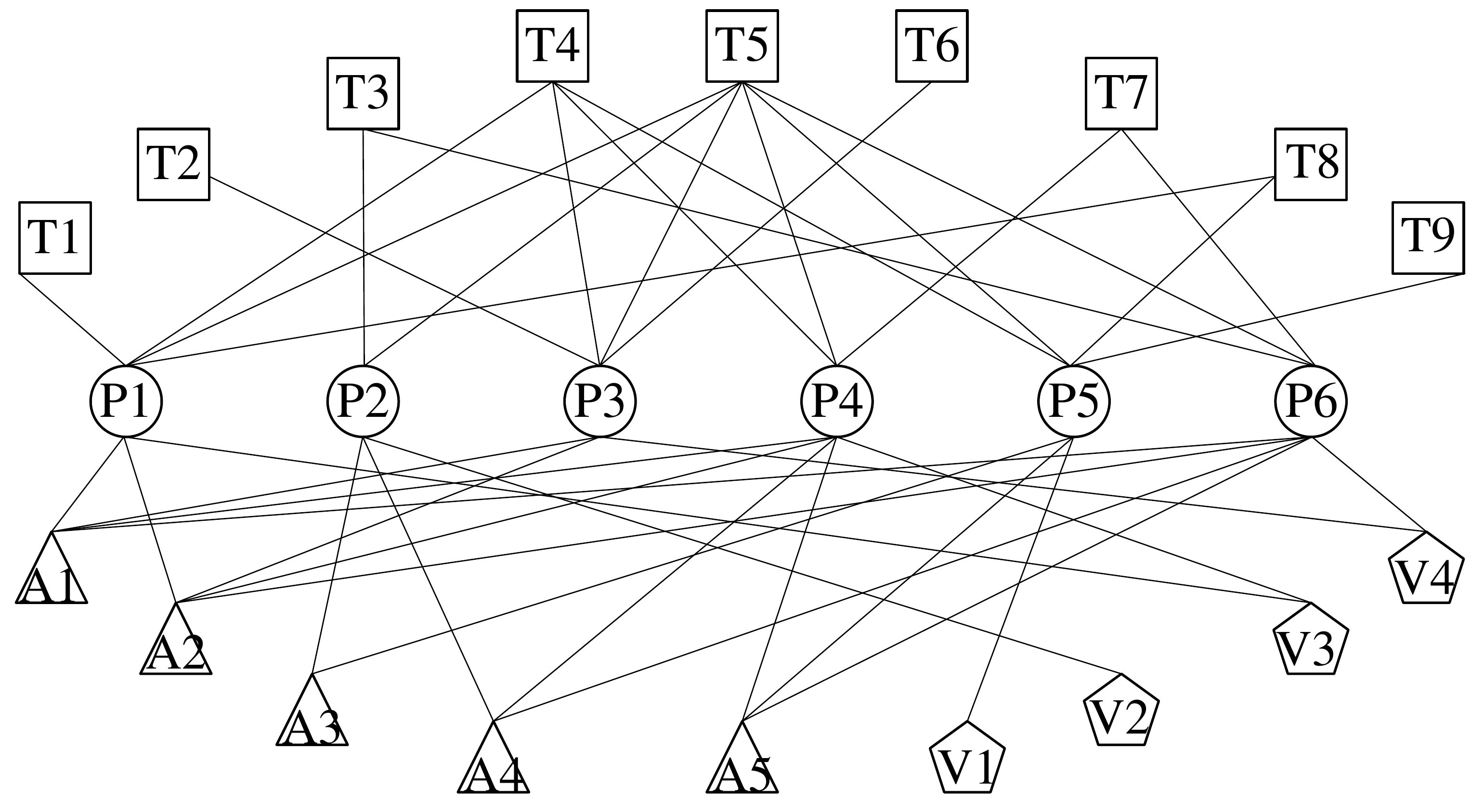}
  \caption{A Toy Bibliographic Information Network.
  T1, T2, T3, T4, T5, T6, T7, T8 and T9 respectively stand for terms `NetworkSchema', `RelationStrength', `Similarity', `Clustering', `HIN', `Attribute', `MetaPath', `Ranking', `NetworkSchema'.
  P1, P2, P3, P4, P5 and P6 respectively stand for papers `NetClust', `HeteSim', `GenClus', `PathSelClus', `HeProjI', `PathSim'.
  A1, A2, A3, A4 and A5 respectively stand for authors `Yizhou Sun', `Jiawei Han', `Chuan Shi', `Philip S. Yu'. `Xifeng Yan'.
  V1, V2, V3 and V4 respectively stand for venues `CIKM', `TKDE', `SIGKDD', `VLDB'.}
  \label{ExampleHins}
\end{figure}

Information network analysis attracts many researchers' attention in the field of data mining because many real systems,
e.g. bibliographic information database and biological systems, can be modeled as information networks.
These networks have common characteristics: they are composed of multi-typed and interconnected objects.
This kind of information networks is usually called \textbf{Heterogeneous Information Networks (HIN)}.
Fig. \ref{ExampleHins} shows a toy bibliographic information network with four actual object types \textit{Author} ($A$)
in the shape of triangles, \textit{Paper} ($P$) in the shape of circles,
\textit{Venue} ($V$) in the shape of pentagons and \textit{Term} ($T$) in the shape of squares.
The type $P$ has six instances: P:HeteSim \cite{SKHYW:2014}, P:HeProjI \cite{SWLYW:2014},
P:GenClus \cite{SAH:2012}, P:PathSelClus \cite{SNHYYY:2012}, P:PathSim \cite{SHYYW:2011}, P:NetClus \cite{SYH:2009}.
Each paper has its author(s), a venue and its related terms.
Hence, it contains three types of links: $P\leftrightarrow A$, $P\leftrightarrow V$ and $P\leftrightarrow T$.

In a \textbf{HIN}, a fundamental problem is to measure the similarities between objects using structural and semantic information.
All the off-the-shelf similarities in \textbf{HIN} are based on user-specified meta paths,
for example $PathSim$ \cite{SHYYW:2011} and Biased Path Constrained Random Walk ($BPCRW$) \cite{LC:2010b, LC:2010a}.
According to the literature \cite{HZCSML:2016}, meta paths can only capture biased and relatively simple semantics.
Therefore, the authors proposed a more complicated structure called meta structure, and defined
the meta structure based similarity using the compressed-ETree, called the Biased Structure Constrained Subgraph Expansion ($BSCSE$).
However, the meta structure needs to be specified in advance as well.

It is really very difficult for users to specify meta paths or meta structures.
For example, there are ten object types (Gene, Gene Ontology, Tissue, Chemical Compound, Side Effect,
Substructure, Chemical Ontology, Pathway, Disease, Gene Family)
and eleven link types in a complete biological information network \cite{CDJWZ:2010,FDSCSB:2016}.
Obviously, users hardly know how to choose appropriate meta paths or meta structures.
In addition, different meta paths and meta structures may have different effects on the similarities between objects.
This makes users more difficult to select appropriate meta paths or meta structures.

To alleviate users' burden, we propose an automatically-constructed schematic structure
called \textbf{S}tratified \textbf{M}eta \textbf{S}tructure (\textbf{SMS}).
It needs not to be specified in advance, and combines many meta paths and meta structure.
This ensures that
(1) Users need not to follow with interest the structure of the network schema of the input \textbf{HIN};
(2) Rich semantics can still be captured.
We are inspired by the tree-walk proposed in \cite{NPEKK:2011}.
The structure of a tree-walk is constructed by repetitively visiting nodes in the input graph.
This idea can be employed here.
As a result, we devise the stratified meta structure,
which is essentially a directed acyclic graph consisting of the object types with different layer labels.
It can be automatically constructed via repetitively visiting the object types on the network schema.
In the process of the construction, we discover the \textbf{SMS} consists of many basic substructures and recurrent substructures,
see section \ref{subsec:similaritymeasure}.
These basic substructures and recurrent substructures essentially represent specific relations.
The \textbf{SMS} as a composite structure is therefore a composite relation. This is why the \textbf{SMS} can capture rich semantics.

After obtaining the \textbf{SMS}, the next step is to formalize its rich semantics.
For meta structures, the compressed-ETree is used to formalize its semantics.
However, it cannot be used here, because \textbf{SMS} contains an infinite number of meta structures.
The semantics contained in meta paths are usually formalized by its commuting matrices.
In essence, the meta structures have the same nature as the meta paths, because they all have hierarchical structures.
So, we define commuting matrices of meta structures by virtue of cartesian product in section \ref{subsec:metapath_metastructure},
and further define commuting matrix of the \textbf{SMS} by reasonably combining the infinite number of the commuting matrices of meta structures.
The proposed metric, $SMSS$, is defined by the commuting matrix of the \textbf{SMS}.
Experimental evaluations suggest that $SMSS$ on the whole outperforms the baselines $PathSim$, $BPCRW$ and $BSCSE$ in terms of ranking quality and clustering quality.

\textbf{The main contributions} are summarized as follows.
\begin{enumerate}[1)]
  \item We propose the stratified meta structure with rich semantics, which can be constructed automatically,
    and define a stratified meta structure similarity $SMSS$ by virtue of the commuting matrix of the \textbf{SMS};
  \item We define the commuting matrices of meta structures by virtue of cartesian product, and use them to compactly re-formulate $BSCSE$;
  \item We conduct experiments for evaluating the performance of the proposed metric $SMSS$.
    The proposed metric on the whole outperforms the baselines in terms of ranking quality and clustering quality.
\end{enumerate}

The rest of the paper is organized as follows. Section \ref{sec:relatedwork} introduces related works.
Section \ref{sec:preliminaries} provides some preliminaries on \textbf{HINs}.
Section \ref{sec:deepmetastructure} introduces the definition of $SMSS$.
The experimental evaluations are introduced in section \ref{sec:expriment}. The conclusion is introduced in section \ref{sec:conclusion}.

\section{Related Work}\label{sec:relatedwork}

To the best of our knowledge, Sun et al. \cite{SYH:2009,SHZYCW:2009,SH:2012} first proposed the definition of the \textbf{HIN}
and studied ranking-based clustering in \textbf{HINs}.
Shi et al. \cite{SLZSY:2017} gave a comprehensive summarization of research topics on \textbf{HINs} including similarity measure, clustering,
link prediction, ranking, recommendation, information fusion and classification etc.
Article \cite{HeteClass:2017} proposed a novel meta-path based framework called HeteClass for transductive classification of target type objects.
This framework can explore the network schema of the input \textbf{HIN} and incorporate the expert's knowledge to generate a collection of meta paths.
Below, we summarize related works on similarity measures in information networks.

For similarity measures in homogeneous information networks,
literature \cite{JW:2002} proposed a general similarity measure $SimRank$ combining the link information, which thought two similar objects must relate to similar objects.
Literature \cite{JW:2003} evaluated the similarities of objects by a random walk model with restart. Article \cite{LinkPred:2017} lists many state-of-the-art similarities in homogeneous information networks:
(1) Local Approaches:
e.g. Common Neighbors ($CN$), Adamic-Adar Index ($AA$), Resource Allocation Index ($RA$),
Resource Allocation based on Common Neighbor Interactions ($RA-CNI$),
Preferential Attachment Index ($PA$), Jaccard Index ($JA$), Salton Index ($SA$), Sorensen Index ($SO$), Hub Promoted Index ($HPI$),
Hub Depressed Index ($HDI$), Local Leicht-Holme-Newman Index ($LLHN$), Individual Attraction index ($IA$), Mutual Information ($MI$),
Local Naive Bayes ($LNB$), CAR-Based Indices ($CAR$), Functional Similarity Weight ($FSW$), Local Interacting Score ($LIT$);
(2) Global Approaches:
Negated Shortest Path ($NSP$), Katz Index ($KI$), Global Leicht-Holme-Newman Index ($GLHN$), Random Walks (RA),
Random Walks with Restart ($RWR$), Flow Propagation ($FP$), Maximal Entropy Random Walk ($MERW$),
Pseudo-inverse of the Laplacian Matrix ($PLM$), Average Commute Time ($ACT$), Random Forest Kernel Index ($RFK$), The Blondel index ($BI$),
(3) Quasi-Local Approaches:
Local Path Index (HPI), Local Random Walks (LRW), Superposed Random Walks (SRW),
Third-Order Resource Allocation Based on Common Neighbor Interactions ($ORA-CNI$), FriendLink ($FL$), PropFlow Predictor ($PFP$).

For Similarity measures in heterogeneous information networks,
Sun \cite{SHYYW:2011} proposed a meta path based similarity measure in \textbf{HINs}, called $PathSim$.
Lao and Cohen \cite{LC:2010b, LC:2010a} studied the problem of measuring the entity similarity in labeled directed graphs, and defined a Biased Path Constrained Random Walk ($BPCRW$) model.
It can be applied to \textbf{HINs}. Huang et al. \cite{HZCSML:2016} proposed a similarity $BSCSE$, which can capture more complex semantics.
Shi et al. \cite{SKHYW:2014} proposed a relevance measure $HeteSim$ which can be used to evaluate the relatedness of two object with different types.
For a user-specified meta path, $HeteSim$ is based on the pairwise random walk from its two endpoints to its center.
Xiong et al. \cite{XZY:2015} studied the problem of finding the $\text{top-}k$ similar object pairs by virtue of locality sensitive hashing.
Zhu et al. \cite{SEMATCH:2017} proposed an integrated framework for the development, evaluation and application of semantic similarity for knowledge graphs
which can be viewed as complicated heterogeneous information networks. This framework included many similarity tools and allowed users to compute semantic similarities.
In the article \cite{ForwardBYW:2017}, the authors studied the similarity search problem in social and knowledge networks, and proposed a dual perspective similarity metric called Forward Backward Similarity.


\section{Preliminaries}\label{sec:preliminaries}
In this section, we introduce some important concepts related to \textbf{HINs} including network schema in subsection \ref{subsec:HINModel}, meta paths and meta structures in \ref{subsec:metapath_metastructure}.

\subsection{HIN Definition}\label{subsec:HINModel}
As defined in article \cite{SWLYW:2014}, an \textbf{information network} is essentially a directed graph $G=(V,E,\mathcal{A},\mathcal{R})$.
$V$ and $E$ respectively denote its sets of objects and links,
and $\mathcal{A}$ and $\mathcal{R}$ respectively denote its sets of object types and link types.
Map $\phi: V\rightarrow\mathcal{A}$ denotes the object type $\phi(v)$ of object $v\in V$.
That is to say, each object in $V$ belongs to a specific object type.
Similarly, map $\psi: E\rightarrow\mathcal{R}$ represents the link type $\psi(e)$ of link $e\in E$, i.e. each link in $E$ belongs to a specific link type.
In essence, $\psi(e)$ contains some semantic because it is a relation from the source object type to the target object type.
If two links belong to the same link type, they share the same starting object type as well as the ending object type.
$G$ is called a \textbf{Heterogeneous Information Network} if $|\mathcal{A}|>1$ or $|\mathcal{R}|>1$.
Otherwise, it is called a homogeneous information network.

\begin{figure}[htb]
  \centering
  \includegraphics[width=0.7\textwidth]{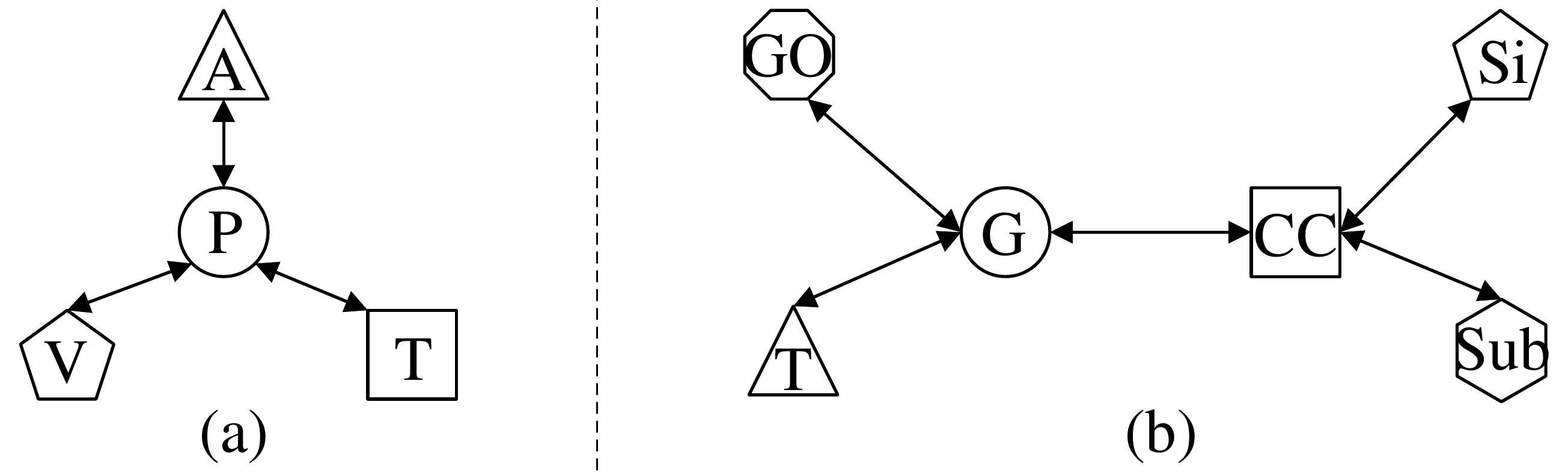}
  \caption{(a) Bibliographic network schema. (b) Biological network schema.}
  \label{exampleNetworkSchema}
\end{figure}

For each \textbf{HIN}, there is a meta-level description for it, called its \textbf{network schema}.
Specifically, the \textbf{Network schema} $\mathcal{T}_G=(\mathcal{A},\mathcal{R})$ \cite{SWLYW:2014} of $G$ is a directed graph consisting of the object types in $\mathcal{A}$
and the link types in $\mathcal{R}$. Fig. \ref{exampleNetworkSchema}(a) shows the network schema for the \textbf{HIN}
in Fig. \ref{ExampleHins}.
In this paper, we also use a biological information network consisting of six object types, i.e. \textit{Gene} ($G$), \textit{Tissue} ($T$), \textit{GeneOntology} ($GO$),
\textit{ChemicalCompound} ($CC$), \textit{Substructure} ($Sub$) and \textit{SideEffect} ($Si$), and five link types, i.e.
\textit{GO$\leftrightarrow$G}, \textit{T$\leftrightarrow$G}, \textit{G$\leftrightarrow$CC}, \textit{CC$\leftrightarrow$Si}, \textit{CC$\leftrightarrow$Sub}.
Its network schema is shown in \ref{exampleNetworkSchema}(b).

\subsection{Schematic Structures}\label{subsec:metapath_metastructure}
Up to now, there are two kinds of schematic structures (meta paths and meta structures) for the network schema of the \textbf{HIN}. All of them carries some semantics.
It is noteworthy that these two kinds of schematic structures must be specified by users when using them to measure the similarities between objects.

\textbf{Meta path} \cite{SHYYW:2011} is essentially an alternate sequence of object types and link types,
i.e. $O_1\xrightarrow{R_1}O_2\xrightarrow{R_2}\cdots\xrightarrow{R_{l-2}}O_{l-1}\xrightarrow{R_{l-1}}O_l$, $O_i\in\mathcal{A}, i=1,\cdots,l$ and $R_j\in\mathcal{R}, j=1,\cdots,l-1$.
Note that $R_j$ is a link type starting from  $O_j$ to $O_{j+1}, j=1,\cdots,l-1$.
In essence, the meta path contains some composite semantic because it represents a composite relation $R_1\circ R_2\circ\cdots\circ R_{l-1}$.
Unless stated otherwise, the meta path $\mathcal{P}=O_1\xrightarrow{R_1}O_2\xrightarrow{R_2}\cdots\xrightarrow{R_{l-2}}O_{l-1}\xrightarrow{R_{l-1}}O_l$
can be compactly denoted as $(O_1,O_2,\cdots,O_{l-1},O_l)$.
There are some useful concepts related to the meta path in literature \cite{SHYYW:2011}, i.e. \textbf{length} of $\mathcal{P}$, \textbf{path instance} following $\mathcal{P}$,
\textbf{reverse meta path} of $\mathcal{P}$, \textbf{symmetric meta path} and \textbf{commuting matrix} $\mathcal{M}_{\mathcal{P}}$ of $\mathcal{P}$.
For example, Fig. \ref{exampleMetaPath}(a,b,c) show three meta paths in the network schema shown in Fig. \ref{exampleNetworkSchema}(a).
They can be compactly denoted as $(A,P,A)$, $(A,P,V,P,A)$ and $(A,P,T,P,A)$. They can express different semantics.
$(A,P,A)$ expresses ``Two authors cooperate on a paper.''
$(A,P,V,P,A)$ express ``Two authors publish their papers in the same venue.''
$(A,P,T,P,A)$ express ``Two authors publish their papers containing the same terms.''

\begin{figure}[htb]
  \centering
  \includegraphics[width=0.7\textwidth]{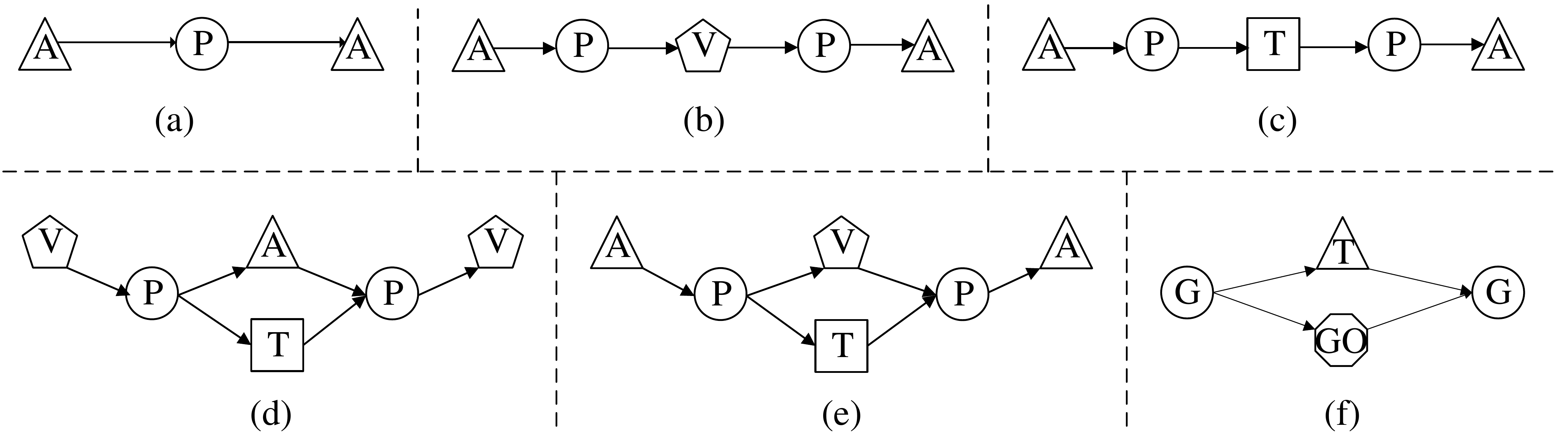}
  \caption{Some Meta Paths and meta structures.}
  \label{exampleMetaPath}
\end{figure}

\textbf{Meta structure} \cite{HZCSML:2016} $\mathcal{S}=(\mathcal{V}_{\mathcal{S}},\mathcal{E}_{\mathcal{S}},T_s,T_t)$
is essentially a directed acyclic graph
with a single source object type $T_s$ and a single target object type $T_t$. $\mathcal{V}_{\mathcal{S}}$
is a set of object types, and $\mathcal{E}_{\mathcal{S}}$ is a set of link types.
Fig. \ref{exampleMetaPath}(d,e) show two kinds of meta structures for the network schema shown in Fig. \ref{exampleNetworkSchema}(a).
All of them can be compactly denoted as $(V,P,(A,T),P,V)$ and $(A,P,(V,T),P,A)$.
Fig. \ref{exampleMetaPath}(f) shows a meta structure for the network schema shown in Fig. \ref{exampleNetworkSchema}(b).
It can be compactly denoted $(G,(GO,T),G)$.
Meta Structure $(V,P,(A,T),P,V)$ expresses the more complicated semantic
``Two venues publish papers both containing the same terms and wrote by the same authors.''
Meta Structure $(A,P,(V,T),P,A)$ expresses the more complicated semantic
``Two authors write their papers both containing the same terms and in the same venue.''

Given a meta structure $\mathcal{S}$, we sort its object types with regard to the topological order. Suppose its height is equal to $h_1$.
Let $L_i$ denote the set of object types on the layer $i$, and $C_{L_i}$ denote the cartesian product of the set of objects belonging to different types in $L_i$, $i=0,1,\cdots,h_1-1$.
The relation matrix $W_{L_iL_{i+1}}$ from $C_{L_i}$ to $C_{L_{i+1}}$ is defined as the one whose entries $(s,t)$ are equal to 1
if the $s$-th element $C_{L_h}(s)$ of $C_{L_h}$ is adjacent to the $t$-th one $C_{L_{h+1}}(t)$ of $C_{L_{h+1}}$ in $G$, otherwise 0.
$C_{L_h}(s)$ and $C_{L_{h+1}}(t)$ are adjacent if and only if for any $u\in C_{L_h}(s)$ and $v\in C_{L_{h+1}}(t)$, $u$ and $v$ are adjacent in $G$ if $\phi(u)$ and $\phi(v)$ are adjacent in $\mathcal{T}_G$.
The commuting matrix of $\mathcal{S}$ is defined as
\begin{displaymath}
\mathcal{M}_{\mathcal{S}}=\prod_{i=0}^{h_1-1}W_{L_iL_{i+1}}.
\end{displaymath}
Each entry in $\mathcal{M}_{\mathcal{S}}$ represents the number of instances following $\mathcal{S}$.
The commuting matrix of its reverse is equal to $\mathcal{M}_{\mathcal{S}}^T$..

\begin{figure}[htb]
  \centering
  \includegraphics[width=1.0\textwidth]{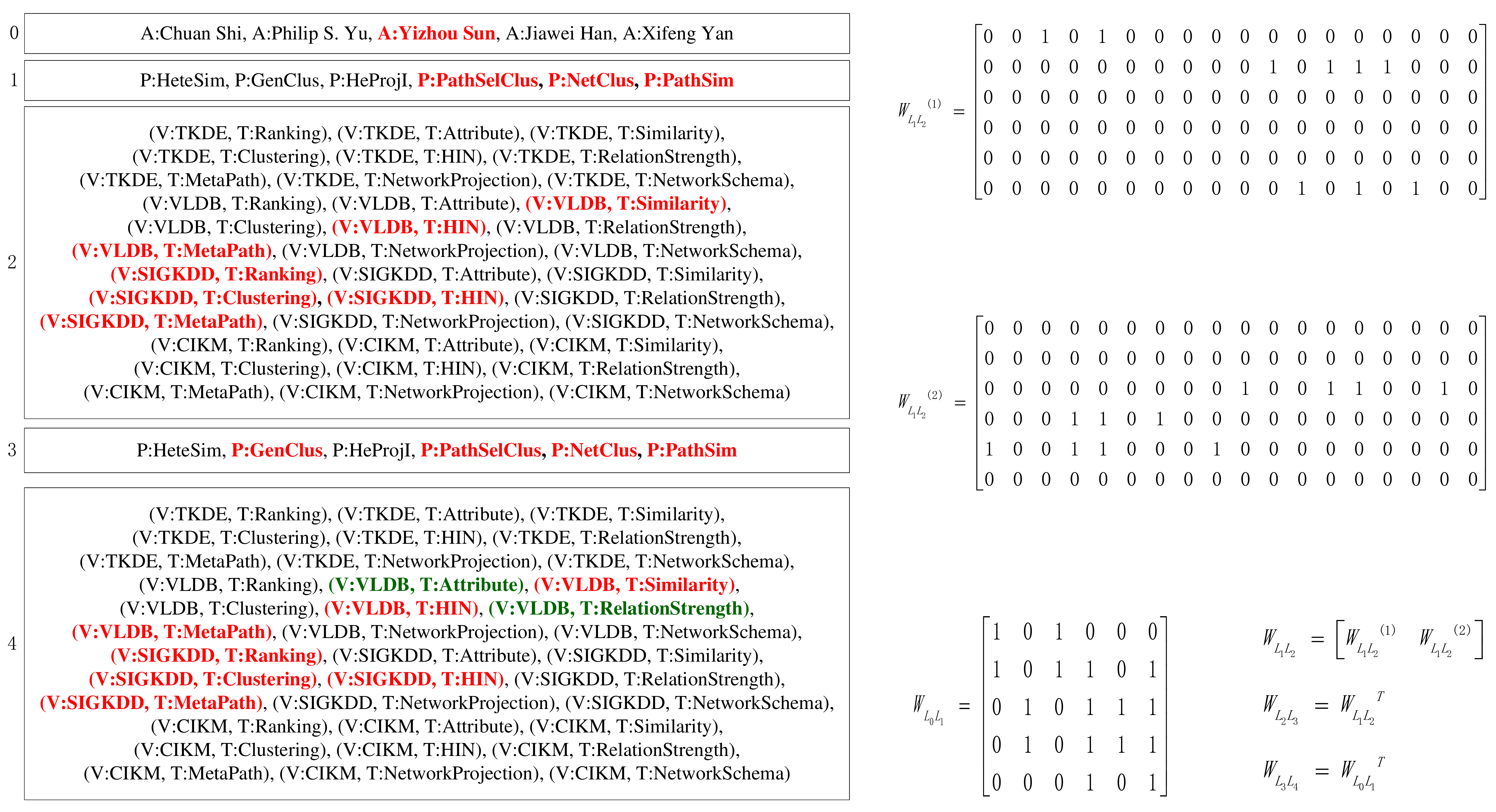}
  \caption{Illustration of $C_{L_i}, i=0,1,2,3,4$ (left hand) and $W_{L_iL_{i+1}},i=0,1,2,3$ (right hand). Because of the space limitation, $W_{L_1L_2}$ is partitioned into two blocks $W_{L_1L_2}^{(1)}$ and $W_{L_1L_2}^{(2)}$.}
  \label{exampleCommutingMatrix}
\end{figure}

Take the \textbf{HIN} shown in Fig. \ref{ExampleHins} as an example. We compute the commuting matrix of the meta structure shown in Fig. \ref{exampleMetaPath}(b).
It has five layers $L_0=\{A\}$, $L_1=\{P\}$, $L_2=\{V,T\}$, $L_3=\{P\}$ and $L_4=\{A\}$.
The $i$-th box on the left-hand side of Fig. \ref{exampleCommutingMatrix} shows the cartesian product $C_{L_i}, i=0,1,2,3,h_1=4$.
Then, we can easily obtain the relation matrices $W_{L_0L_1}$, $W_{L_1L_2}$, $W_{L_2L_3}=W_{L_1L_2}^T$, $W_{L_3L_4}=W_{L_0L_1}^T$ on the right-hand side of Fig. \ref{exampleCommutingMatrix}.
According to the fact that $W_{L_1L_2}(1,3)=1$, we know P:HeteSim is adjacent to (V:TKDE,T:Similarity). In fact, P:HeteSim is published on the V:TKDE and contains the term T:Similarity.
Similarly, $W_{L_1L_2}(1,1)=0$ implies that P:HeteSim is not adjacent to (V:TKDE,T:Ranking). According to the \textbf{HIN} shown in Fig. \ref{ExampleHins}, we know P:HeteSim does not contain the term T:Ranking.
As a result,
\begin{displaymath}
\begin{array}{ll}
\mathcal{M}_{\mathcal{S}} & =W_{L_0L_1}\times W_{L_1L_2}\times W_{L_2L_3}\times W_{L_3L_4} \\
                          & =W_{L_0L_1}\times W_{L_1L_2}\times W_{L_1L_2}^T\times W_{L_0L_1}^T \\
                          & =
\begin{bmatrix}
6 & 6 & 0 & 0 & 0 \\
6 & 18 & 14 & 14 & 8 \\
0 & 14 & 20 & 20 & 11 \\
0 & 14 & 20 & 20 & 11 \\
0 & 8 & 11 & 11 &9
\end{bmatrix}
.
\end{array}
\end{displaymath}

For a given meta structure, its $BSCSE$ with $\alpha=1$ can be expressed by its commuting matrix as well. The following lemma \ref{extension_to_ms} describes this conclusion.
Throughout this paper, we use $\bar{X}=U_X^{-1}X$ to denote its normalized version, where $U_X$ is a diagonal matrix whose nonzero entries are equal to the row sum of $X$.
\begin{lemma}\label{extension_to_ms}
 Given a meta structure $\mathcal{S}$, suppose that $L_i$ denotes the set of object types on its $i$-th layer. When $\alpha=1$,
\begin{displaymath}
BSCSE(o_s,o_t|\mathcal{S},h,\alpha)=\mathcal{M}_{\mathcal{S}}^h(o_s,o_t),
\end{displaymath}
where $\mathcal{M}_{\mathcal{S}}^h=\bar{W}_{L_0L1}\times\bar{W}_{L_1L_2}\times\cdots\times\bar{W}_{L_{h-1}L_{h}}$.
\end{lemma}
\begin{proof}
We prove the lemma by induction on $h\geq 1$.

\textbf{Initial Step}. Obviously, $C_{L_0}=\{o_s\}$. When $h=1$, $\mathcal{M}_{\mathcal{S}}^1=\bar{W}_{L_0L_1}$. Assume there are $d_s$ different object tuples in $C_{L_1}$ adjacent to $o_s$,
denoted as $o_{11},o_{12},\cdots,o_{1d_s}$.
According to the definition of $w(v)$ \cite{HZCSML:2016},
$w(o_{1i})=\frac{1}{d_s}, i=1,2,\cdots,d_s$. Obviously, $\bar{W}_{L_0L_1}(o_s,o_{1i})=\frac{1}{d_s}, i=1,2,\cdots,d_s$.
Therefore, we have $BSCSE(o_s,o_{1i}|\mathcal{S},1,\alpha)=\mathcal{M}_{\mathcal{S}}^1(o_s,o_{1i}), i=1,2,\cdots,d_s$.

\textbf{Inductive Step}. Assume the conclusion holds for $h$. Below, we prove it also holds for $h+1$.
Obviously,
\begin{equation}\label{xxxxx}
\begin{array}{ll}
\mathcal{M}_{\mathcal{S}}^{h+1} & =\mathcal{M}_{\mathcal{S}}^h\bar{W}_{L_hL_{h+1}} \\
         & =\left[\sum_{k=1}^lx_ky_{k1}, \cdots, \sum_{k=1}^lx_ky_{km}\right],
\end{array}
\end{equation}
where $\mathcal{M}_{\mathcal{S}}^h=\left[x_1,x_2,\cdots,x_l\right]$, and $\bar{W}_{L_hL_{h+1}}=(y_{ij})_{l\times m}$.
Note in particular that $x_i=BSCSE(o_s,o_i|\mathcal{S},h,\alpha)$, where $i=1,\cdots,l$, and $y_{ij}$, where $i=1,\cdots,l,j=1,\cdots,m$, is equal to either 0 or $\frac{1}{|\{y_{ij}\neq 0|j=1,\cdots,m\}|}$.
According to the definition of  $BSCSE$ in literature \cite{HZCSML:2016},
\begin{equation}\label{yyyyy}
BSCSE(o_s,o_j|\mathcal{S},h+1,\alpha)=\sum_{k=1}^lx_ky_{kj}, j=1,\cdots,m.
\end{equation}
Combining formulas \ref{xxxxx} and \ref{yyyyy}, we have
\begin{displaymath}
BSCSE(o_s,o_j|\mathcal{S},h+1,\alpha)=\mathcal{M}_{\mathcal{S}}^{h+1}(o_s,o_j),
\end{displaymath}
where $j=1,\cdots,m$. The conclusion holds for $h+1$.
\end{proof}

In this paper, we aim to define a similarity measure in \textbf{HINs}, which does not depend on any pre-specified schematic structures.
This is a reasonable restriction because specifying meta paths or meta structures is a cumbersome job.

\section{Stratified Meta Structure Based Similarity}\label{sec:deepmetastructure}

In this section, we  define the stratified meta structure based similarity measure in \textbf{HINs}.
Firstly, we give the architecture of the stratified meta structure in section \ref{subsec:metastructure}.
Secondly, we formally define the similarity based on the stratified meta structure in section \ref{subsec:similaritymeasure}.
At last, we describe the pseudo-code of computing the similarity.

\subsection{Stratified Meta Structure}\label{subsec:metastructure}

A \textbf{Stratified Meta Structure} is essentially a directed acyclic graph consisting of object types with different layer labels.
Its salient advantage is that it can be automatically constructed by repetitively visiting object types in the process of traversing the network schema.
Given a \textbf{HIN} $G$, we first extract its network schema $\mathcal{T}_G$, and then select a source object type and a target object type.
Unless stated otherwise, the source object type is the same as the target one.
\textbf{The construction rule of SMS $\mathcal{D}_G$ of $G$} is described as follows.
The source object type is placed on the 0-th layer. The object types on the layer $l=1,2,\cdots,+\infty$ are composed of the neighbors of the object types on the layer $l-1$ in $\mathcal{T}_G$.
The adjacent object types are linked by an arrow pointing from the $(l-1)$-th layer down to the $l$-th layer.
Note in particular that once we get the target object type on the layer $l\geq 1$, delete its outgoing links, i.e. the ones starting from it down to the object types on the $(l+1)$-th layer.
Repeating the above process, we obtain the \textbf{SMS} $\mathcal{D}_G$.

\begin{figure}[htb]
  \centering
  \includegraphics[width=0.7\textwidth]{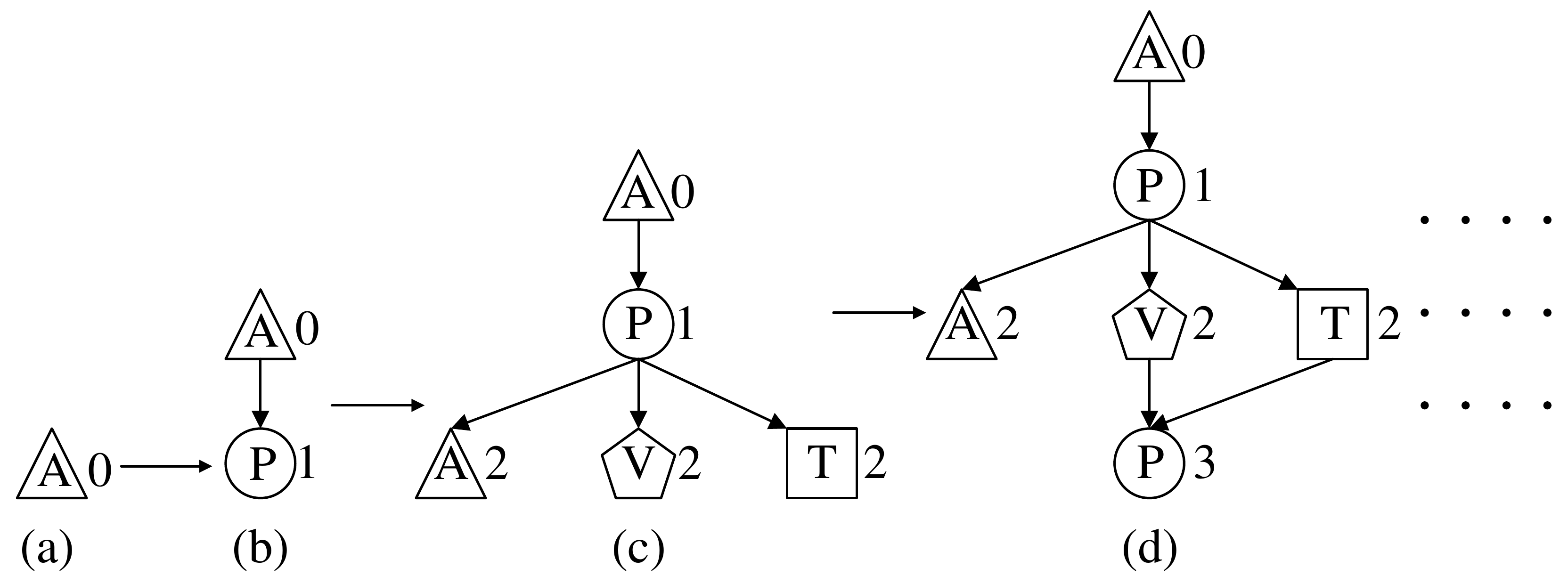}
  \caption{The construction of the \textbf{SMS} of the toy bibliographic information network. The numbers near nodes stand for their layer labels}
  \label{construction_DeepMS}
\end{figure}

\begin{figure}[htb]
  \centering
  \includegraphics[width=0.7\textwidth]{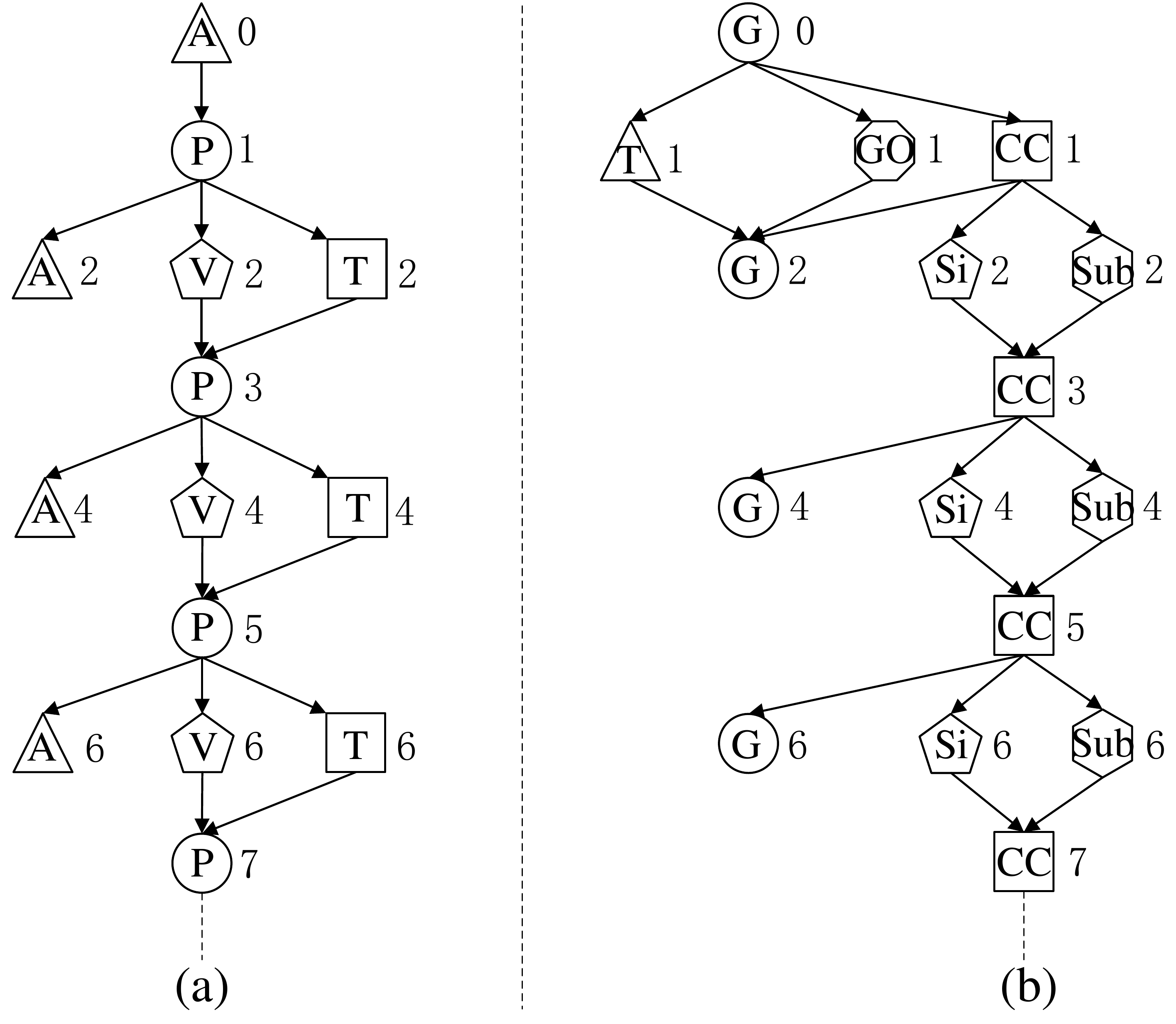}
  \caption{Two kinds of \textbf{SMS}. The numbers near nodes stand for the layer labels. In this figure, the $i$-th layer, $i>8$, are omitted because of space limitation.}
  \label{ExampleDeepMetaStructure}
\end{figure}

Fig. \ref{ExampleDeepMetaStructure}(a) shows the \textbf{SMS} of the network schema shown in Fig. \ref{exampleNetworkSchema}(a).
It can be constructed as shown in Fig. \ref{construction_DeepMS}.
$A$ is both the source and target object type. Firstly, $A$ labelled as $A\_0$ is placed on the 0-th layer, see Fig. \ref{construction_DeepMS}(a).
$P$ is placed on the 1-st layer and labelled as $P\_1$, because $P$ is the only neighbor of $A$ in the network schema shown in Fig. \ref{exampleNetworkSchema}(a), see Fig. \ref{construction_DeepMS}(b).
$A$, $V$ and $T$, respectively labelled as $A\_2$, $V\_2$ and $T\_2$, are placed on the 3-rd layer, because they are the neighbors of $P$, see Fig. \ref{construction_DeepMS}(c).
Similarly, $P$, labelled as $P\_3$, is again placed on the 4-th layer, because it is the neighbor of both $V$ and $T$, see Fig. \ref{construction_DeepMS}(d). At this time, $P$ is visited again.
Note in particular that the link from $A\_2$ down to $P\_3$ is deleted, because $A$ is the target object type.
Repeating the above procedure, we obtain the \textbf{SMS} shown in Fig. \ref{ExampleDeepMetaStructure}(a).
Fig. \ref{ExampleDeepMetaStructure}(b) shows the \textbf{SMS} of the network schema shown in Fig. \ref{exampleNetworkSchema}(b). Gene is both the source and target object type.
It can be constructed as similarly as \ref{ExampleDeepMetaStructure}(a). It is worth noting that $T\_1$ and $GO\_1$ are only placed on the 1-st layer,
because their degrees in the network schema shown in Fig. \ref{exampleNetworkSchema}(b) are equal to 1.

Below, we give some properties of \textbf{SMS} via lemma \ref{properties_dms}.
Given a \textbf{SMS} $\mathcal{D}_G$, we sort its object types in the topological order.
Let $L_h$ denote the set of object types with the layer label $h$ except the target object type.
Let $C_{L_h}$ denote the cartesian product of the set of objects belonging to different types in $L_h$, $h=0,1,\cdots,+\infty$.
\textbf{The relation matrix} $\mathbf{W_{L_hL_{h+1}}}$ \textbf{from} $\mathbf{C_{L_h}}$ \textbf{to} $\mathbf{C_{L_{h+1}}}$ is defined as similarly as
the commuting matrices of meta structures defined in section \ref{subsec:metapath_metastructure}.
A substructure consisting of three layers $h',h'+1,h'+2$ in $\mathcal{D}_G$
is recurrent if and only if $L_h=L_{h+2}$ for $h=h',h'+2,\cdots,+\infty$.
Let $h_0$ denote the height of the spanning tree of $\mathcal{T}_G$ yielded by Breadth-First Search (\textbf{BFS}).
Without loss of generality, we assume the source object type in the network schema does not contain self-loops.
If the source object type has a self-loop, we assign two roles to it: target object type and intermediate object type.
The first role is to treat it as the target object type, and the second is to treat it as an non-source and non-target object type.

\begin{note}
\label{column_row_removing}
In the process of computing $W_{L_iL_{i+1}}$, we need to visit all the elements in $C_{L_i}$.
In practice, there are many all-zero columns in $W_{L_iL_{i+1}}$.
We should remove all the all-zero columns in $W_{L_iL_{i+1}}$ and the corresponding rows in $W_{L_{i+1}L_{i+2}}$.
Removing the $r$-th column of $W_{L_{i}L_{i+1}}$ implies there are no links between the $r$-th object in $C_{L_{i+1}}$ and any object in $C_{L_i}$.
Therefore, it is unnecessary to consider the links between it and the objects in $C_{L_{i+2}}$.
\end{note}

\begin{note}
\label{locality}
In the process of computing $W_{L_iL_{i+1}}$, the objects except the source $o_s$ can be removed from $C_{L_0}$.
At this time, we only need to consider the elements adjacent to $o_s$ in $C_{L_{1}}$, and the others are removed from $C_{L_{1}}$.
In general, suppose $F$ is the set of considered elements in $C_{L_i}$. That implies the elements in $C_{L_i}-F$ are removed from $C_{L_i}$.
In $C_{L_{i+1}}$, we only need to consider the elements in $\cup_{f\in F}N_{L_{i+1}}(f)$ where $N_{L_{i+1}}(f)$ denotes the set of elements in $C_{L_{i+1}}$ adjacent to $f$.
The others are removed from $C_{L_{i+1}}$.
\end{note}

\begin{lemma}\label{properties_dms}
Assume the source object type is the same as the target one, and the source object type in $\mathcal{T}_G$ does not contain self-loops.
The \textbf{SMS} $\mathcal{D}_G$ has the properties:
\begin{enumerate}[1)]
  \item The target object type lies on the $(2i)$-th layer, $i\in\mathbb{N}^+$.
  \item If we walk up from the target object type on the layer $h=2,4,\cdots$ to the source object type along the parents of object types,
    then we obtain a symmetric meta structure, denoted as $\mathcal{P}_h=\left(L_0L_1\cdots L_{\frac{h}{2}}\cdots L_1L_0\right)$.
  \item For any $h\geq h_0$, $L_h=L_{h+2}$.
  \item The substructure consisting of the object types except the target object type on the layers $h_0,h_0+1,h_0+2$ always recurrently appear in the \textbf{SMS}.
  \item For $h=2,4,\cdots,\infty$, the meta structure $\mathcal{P}_h$ contains $n(h)$ recurrent structures, where
    \begin{eqnarray}\label{n_h}
    n(h)=
    \begin{cases}
    \frac{h}{2}-h_0 & h\geq 2h_0 \\
    0 & otherwise
    \end{cases}
    .
    \end{eqnarray}
\end{enumerate}
\end{lemma}
\begin{proof}
1) The conclusion holds obviously according to the construction rule of $\mathcal{D}_G$ and our assumptions.

2) According to property 1, the target object type with different layer labels lies on the even number layer. Thus, even is the height of the meta structure obtained by walking up from the target object type on layer $h$.
According to the construction criteria of the \textbf{SMS}, the meta structure is symmetric with respect to the layer $\frac{h}{2}$.

3) For any $t\in L_h, h\geq h_0$, it must be adjacent to an object type on the layer $h+1$. Obviously, $t\in L_{h+2}$ according to the construction rule of $\mathcal{D}_G$.
Thus, we have $L_h\subseteq L_{h+2}$. Similarly, $L_{h+2}\subseteq L_h$. Therefore, $L_h=L_{h+2}$.

4) This property obviously holds according to property 3.

5) We prove the lemma by induction on $h$.

\textbf{Initial Step.} When $h\leq 2h_0$, we obviously have $n(h)=0$ according to the construction rule of $\mathcal{D}_G$.

\textbf{Inductive Step.} Assume $n(h)=\frac{h}{2}-h_0$ for the layer $h>2h_0$. Below, we prove that the conclusion holds for $h+2$.
When $h_0\%2=0$, we obtain a new recurrent structure consisting of $L_{h-3},L_{h-2},L_{h-1}$.
When $h_0\%2\neq 0$, we obtain a new recurrent structure consisting of $L_{h-2},L_{h-1},L_{h}$.
Therefore, we have
\begin{displaymath}
\begin{array}{rl}
n(h) & =\frac{h}{2}-h_0+1 \\
  & =\frac{h+2}{2}-h_0.
\end{array}
\end{displaymath}
The conclusion holds.
\end{proof}

\begin{figure}[htb]
  \centering
  \includegraphics[width=0.7\textwidth]{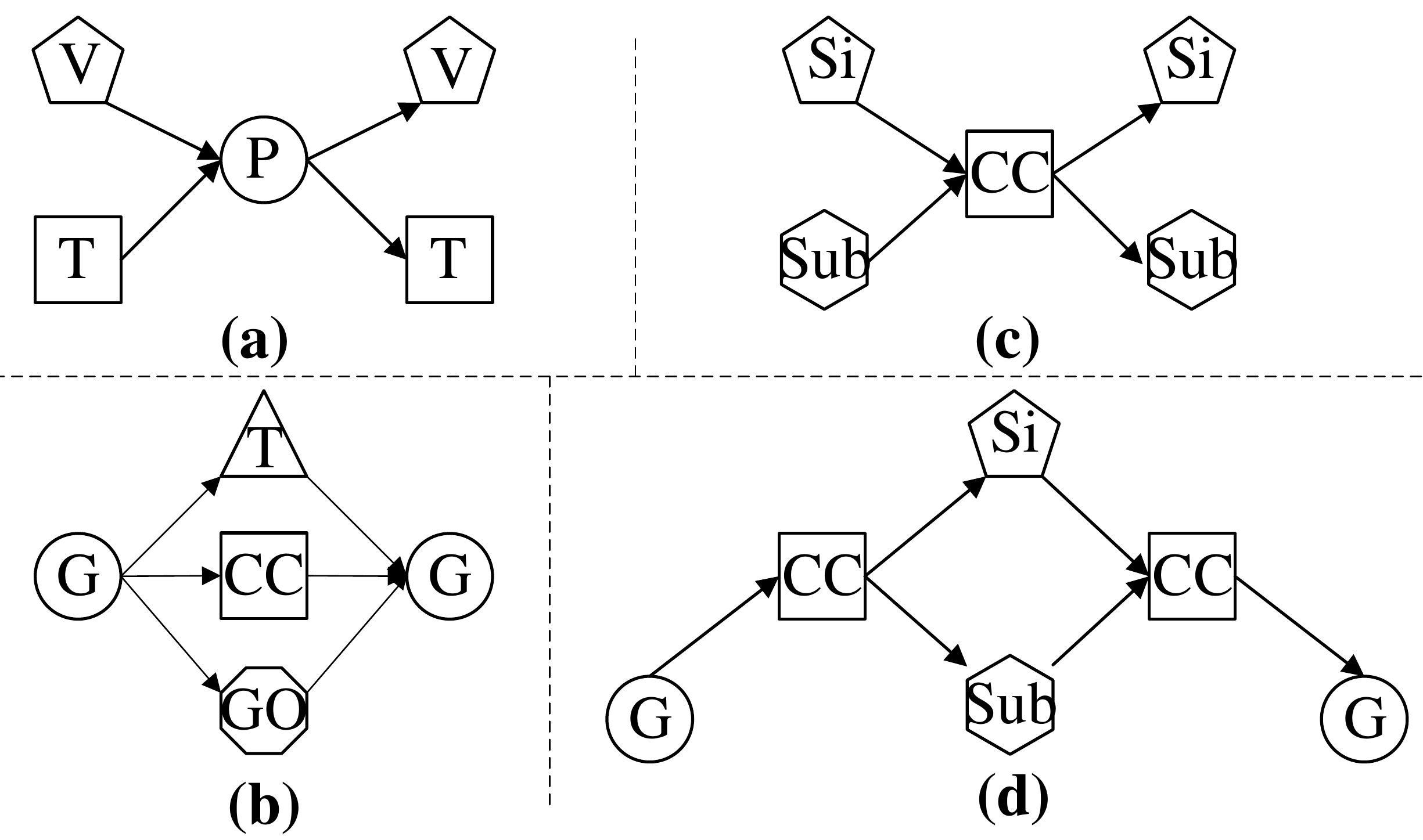}
  \caption{Recurrent structures and meta structures.}
  \label{ExampleMetaStructure}
\end{figure}

According to property 2 in lemma \ref{properties_dms}, \textbf{SMS} is essentially composed of an infinite number of meta paths and meta structures.
For example, the \textbf{SMS} shown in Fig. \ref{ExampleDeepMetaStructure}(a) can be obtained
by combining the meta path shown in Fig. \ref{exampleMetaPath}(a), the meta structure shown in Fig. \ref{exampleMetaPath}(b) and
the others with one or more recurrent substructures shown in Fig. \ref{ExampleMetaStructure}(a).
the \textbf{SMS} shown in Fig. \ref{ExampleDeepMetaStructure}(b) can be obtained by combining the meta structures shown in Fig. \ref{ExampleMetaStructure}(b,d)
and the others with one or more recurrent substructures shown in Fig. \ref{ExampleMetaStructure}(c).
It is noteworthy that the meta structure shown in Fig. \ref{ExampleMetaStructure} can be compactly denoted as $(G,CC,(Si,Sub),CC,G)$.

\subsection{Similarity}\label{subsec:similaritymeasure}

\begin{figure}[htb]
  \centering
  \includegraphics[width=0.7\textwidth]{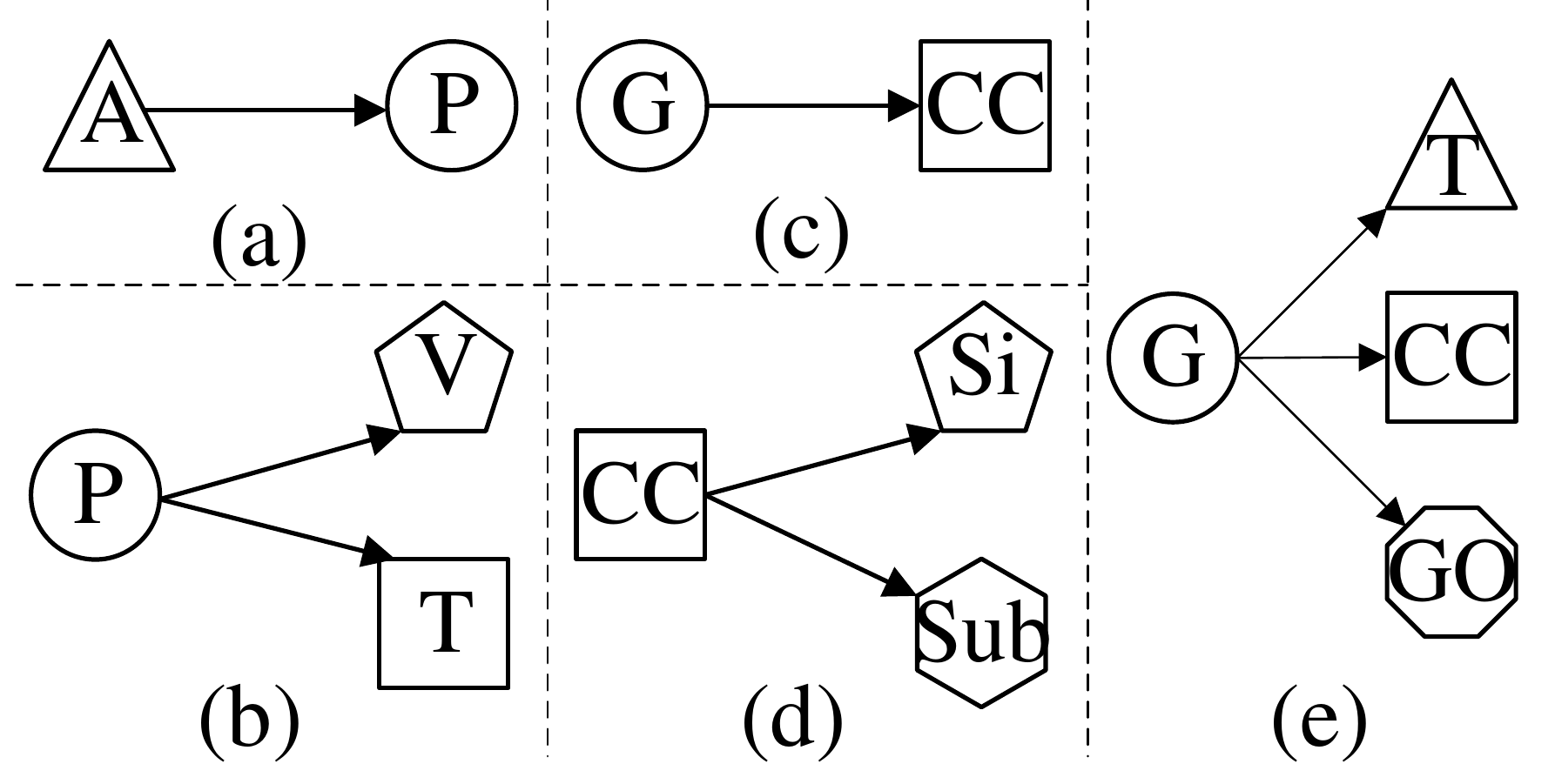}
  \caption{Illustration of basic substructures.}
  \label{basicSubstructures}
\end{figure}

Now, we define the stratified meta structure based similarity by virtue of the commuting matrices of meta paths and meta structures.
The \textbf{SMS} is essentially composed of a recurrent substructure, several basic substructures and their reverses.
The basic substructures are bipartite graphs consisting of the object types in $L_h$ and $L_{h+1}$, $h=0,\cdots,h_0-1$.
The recurrent substructure consists of $L_{h_0},L_{h_0+1},L_{h_0+2}$.
The symmetric meta structures obtained by walking up from the target object type on the layer $h=2,4,\cdots,2h_0$ to the source object type consist of the basic substructures and its reverses.
The symmetric ones obtained by walking from the target object type on the layer $h=2h_0+2,2h_0+4,\cdots,+\infty$ to the source object type consist of one or more recurrent substructures,
the basic substructures and its reverses.
For example, the \textbf{SMS} shown in Fig. \ref{ExampleDeepMetaStructure}(a) can be obtained by combining the recurrent substructure shown in Fig. \ref{ExampleMetaStructure}(a)
and two basic substructures shown in Fig. \ref{basicSubstructures}(a,b).
The \textbf{SMS} shown in Fig. \ref{ExampleDeepMetaStructure}(b) can be obtained by combining the recurrent substructure shown in Fig. \ref{ExampleMetaStructure}(c)
and three basic substructures shown in Fig. \ref{basicSubstructures}(c,d,e).

The \textbf{commuting matrix of a stratified meta structure} is formally defined as the summation of
the commuting matrices of meta paths and meta structures.
Let $\mathcal{M}_{\mathcal{D}}$ denote the commuting matrix of \textbf{SMS} $\mathcal{D}$, and
$\mathcal{S}_h$ denote the meta structure (possibly meta path), which is obtained by walking up
from the target object type on the layer $h$ to the source object type, $h=2, 4,\cdots,\infty$.
Therefore, $\mathcal{M}_{\mathcal{D}}=\sum_{i=1}^{+\infty}\mathcal{M}_{\mathcal{S}_{2i}}$.

As stated previously, $W_{L_hL_{h+1}}$ denotes the relation matrix from $C_{L_h}$ to $C_{L_{h+1}}$.
For the basic substructure consisting of the object types on the layers $h$ and $h+1$,
its relation matrix is just equal to $W_{L_hL_{h+1}}$, $h=0,1,\cdots,h_0-1$.
For the recurrent substructure, its relation matrix is equal to $W_{L_{h_0}L_{h_0+1}}W_{L_{h_0}L_{h_0+1}}^T$
according to property 3 in lemma \ref{properties_dms}.
Below, we show how to compute $\mathcal{M}_{\mathcal{S}_{h}}, h=2,4,\cdots,+\infty$.

\begin{lemma}\label{compute_M_SI}
For any $h=2,4,\cdots,+\infty$, $\mathcal{M}_{\mathcal{S}_{h}}$ can be computed as follows.
Let $R=W_{L_{h_0}L_{h_0+1}}W_{L_{h_0}L_{h_0+1}}^T$.
\begin{enumerate}[1)]
  \item If the degree of the source object type in $\mathcal{T}_G$ is equal to 1, then we have
    \begin{eqnarray}\label{no_degree_1_even}
      \mathcal{M}_{\mathcal{S}_h}=
      \begin{cases}
      U_hU_h^T & h<2h_0 \\
      U_{2h_0}R^{n(h)}U_{2h_0}^T & h\geq 2h_0
      \end{cases}
      ,
    \end{eqnarray}
    where $U_h=\prod_{i=0}^{\frac{h}{2}-1}W_{L_iL_{i+1}}$
  \item If the degree of the source object type in $\mathcal{T}_G$ is larger than 1, then let $L_1'$ denote the set of the object types with degree larger than 1 in the neighbors of the source object type,
    and let $X=W_{L_0L_1'}W_{L_1'L_2}$. When $2<h<2h_0$, let $B_h=X\prod_{i=2}^{\frac{h}{2}-1}W_{L_iL_{i+1}}$.
    We have
    \begin{eqnarray}\label{degree_1_even}
      \mathcal{M}_{\mathcal{S}_h}=
      \begin{cases}
        W_{L_0L_1}W_{L_0L_1}^T & h=2 \\
        B_hB_h^T & 2<h<2h_0 \\
        B_{2h_0}R^{n(h)}B_{2h_0}^T & h\geq 2h_0.
      \end{cases}
    \end{eqnarray}
\end{enumerate}
\end{lemma}
\begin{proof}
For case 1, we prove it by induction on $h$. Case 2 is similar.

\textbf{Initial Step}. When $h=2$, The obtained meta structure (possibly meta path) consists of $\left(L_0L_1L_0\right)$.
Obviously, $\mathcal{M}_{\mathcal{S}_2}=W_{L_0L_1}W_{L_0L_1}^T=A_2A_2^T$

\textbf{Inductive Step}. Assume the conclusion holds for $h$, and the meta structure for $h$ is $\left(L_0L_1\cdots L_{\frac{h}{2}}\cdots L_1L_0\right)$. Now, we prove the conclusion also holds for $h+2$.
According to property 3 in lemma \ref{properties_dms}, the meta structure for $h+2$ is $\left(L_0L_1\cdots L_{\frac{h}{2}}L_{\frac{h+2}{2}}L_{\frac{h}{2}}\cdots L_1L_0\right)$.
When $h<2h_0$, $A_{h+2}=A_{h}\times W_{L_{\frac{h}{2}}L_{\frac{h+2}{2}}}$.
Thus, we have
\begin{displaymath}
\begin{array}{lll}
\mathcal{M}_{\mathcal{S}_{h+2}} & = & \left(\prod_{i=0}^{\frac{h+2}{2}-1}W_{L_iL_{i+1}}\right)\left(\prod_{i=0}^{\frac{h+2}{2}-1}W_{L_iL_{i+1}}\right)^T \\
                                & = & A_{h+2}A_{h+2}^T.
\end{array}
\end{displaymath}
When $h\geq 2h_0$, there are $n(h)$ recurrent substructures according to property 6 in lemma \ref{properties_dms}.
Thereby, the obtained meta structure for $h$ can be denoted as
\begin{displaymath}
\left(L_0\cdots L_{h_0-1}(L_{h_0}L_{h_0+1}L_{h_0})^{n(h)}L_{h_0-1}\cdots L_0\right).
\end{displaymath}
The obtained meta structure for $h+2$ is
\begin{displaymath}
\left(L_0\cdots L_{h_0-1}(L_{h_0}L_{h_0+1}L_{h_0})^{n(h+2)}L_{h_0-1}\cdots L_0\right).
\end{displaymath}
Therefore,
\begin{displaymath}
\begin{array}{lll}
  \mathcal{M}_{\mathcal{S}_{h+2}} & =A_{2h_0}\left(W_{L_{h_0}L_{h_0+1}}W_{L_{h_0}L_{h_0+1}}^T\right)^{n(h)+1}A_{2h_0}^T \\
                                  & =A_{2h_0}\left(W_{L_{h_0}L_{h_0+1}}W_{L_{h_0}L_{h_0+1}}^T\right)^{n(h+2)}A_{2h_0}^T \\
                                  & =A_{2h_0}R^{n(h+2)}A_{2h_0}^T.
\end{array}
\end{displaymath}
So, the conclusion holds when $h+2$.
\end{proof}

Below, we only discuss case 1 in lemma \ref{compute_M_SI}. Case 2 is similar. Obviously,
\begin{displaymath}
\mathcal{M}_{\mathcal{D}_G}=\sum_{h=2}^{2h_0-2}A_hA_h^T+A_{2h_0}\left(\sum_{j=0}^{+\infty}R_1^j\right)A_{2h_0}^T.
\end{displaymath}
The matrix power series $\sum_{j=0}^{+\infty}R_1^j$ may be divergent. In addition, different meta structures in the \textbf{SMS} should also have different weights.
As a result, the normalized version of $\mathcal{M}_{\mathcal{D}}$ is equal to
\begin{displaymath}
\begin{array}{ll}
\bar{\mathcal{M}}_{\mathcal{D}_G}= & \displaystyle\sum_{h=2}^{2h_0-2}w_{\frac{h}{2}-1}\bar{A}_h\bar{A}_h^T \\
                                 & +w_{h_0-1}(1-\lambda)\bar{A}_{2h_0}\left(\sum_{j=0}^{+\infty}\lambda\bar{R}^j\right)\bar{A}_{2h_0}^T,
\end{array}
\end{displaymath}
where $\bar{W}_{L_iL_{i+1}}$ and $\bar{R}$ are respectively the normalized versions of $W_{L_iL_{i+1}}$ and $R$, and $\bar{A}_h=\prod_{i=0}^{\frac{h}{2}-1}\bar{W}_{L_iL_{i+1}}$.
$\lambda\in\left(0,1\right)$ is called \textbf{decaying factor}.
$w_h\in [0,1], h=0,1,\cdots,h_0$, satisfying $\sum_{h=0}^{h_0}w_h=1$, denote the weights of different meta structures.
Obviously, the spectral radius $\rho(\lambda\bar{R})$ of $\lambda\bar{R}$ is less than 1 because $\bar{R}$ is a row random matrix and $0<\lambda<1$.
Let $\mathbf{I}$ denote the identity matrix with the same size as $\bar{R}$.
As a result, we have
\begin{equation}\label{normalized_commuting_matrix_dms}
\begin{array}{lll}
\mathcal{\bar{M}}_{\mathcal{D}_G}= & \displaystyle\sum_{h=2}^{2h_0-2}w_{\frac{h}{2}-1}\bar{A}_h\bar{A}_h^T \\
                                    & +w_{h_0-1}(1-\lambda)\bar{A}_{2h_0}\left(\mathbf{I}-\lambda\bar{R}\right)^{-1}\bar{A}_{2h_0}^T.
\end{array}
\end{equation}
The \textbf{Stratified Meta Structure based Similarity}, $SMSS$, of the source object $o_s$ and the target object $o_t$ is defined as
\begin{equation}\label{dms_sim}
SMSS(o_s,o_t)=\frac{2\times\mathcal{\bar{M}}_{\mathcal{D}_G}(o_s,o_t)}{\mathcal{\bar{M}}_{\mathcal{D}_G}(o_s,o_s)+\mathcal{\bar{M}}_{\mathcal{D}_G}(o_t,o_t)}.
\end{equation}

\begin{note}\label{approximation}
Using note \ref{locality} result in that $C_{L_{h}}\neq C_{L_{h+2}}$ for $h=h_0,h_0+2,\cdots,+\infty$.
As a result, the matrix $R$ defined previously is not square.
To address this issue, the elements in $C_{L_{h+2}}-C_{L_{h}}$ are removed from $C_{L_{h+2}}$. This leads to losing some semantics.
\end{note}

Now, we take the \textbf{HIN} shown in Fig. \ref{ExampleHins} as an example to illustrate note \ref{approximation}.
The object A:Yizhou Sun is selected as the source one. As shown in the first box ($C_{L_0}$) of the left-hand side of Fig. \ref{exampleCommutingMatrix},
A:Yizhou Sun marked as red color is kept, and the others are removed. In the second and third boxes ($C_{L_1}$ and $C_{L_2}$),
the elements marked as red color are kept and the others are removed.
In the fourth box ($C_{L_3}$), P:GenClus in addition to the red elements in $C_{L_1}$ is also kept because it is adjacent to (V:VLDB,T:HIN). Obviously, $C_{L_1}\neq C_{L_3}$.
In the fifth box ($C_{L_4}$), the elements marked green color in addition to the red ones in $C_{L_2}$ should also be kept because they are adjacent to P:GenClus.
Obviously, $C_{L_2}\neq C_{L_4}$. According to the approximation strategy, they are removed from $C_{L_4}$. That means some semantics are lost.

According to notes \ref{locality} and \ref{approximation} some elements in $C_{L_{i}}$ are removed.
$\bar{A}_h, h=2,4,\cdots,2h_0$ and $\bar{R}$ in formula \ref{normalized_commuting_matrix_dms} should be adjusted accordingly.
$W_{L_iL_{i+1}}$ is still used to denote the relation matrix from the renewed $C_{L_i}$ to the renewed $C_{L_{i+1}}$.
In $\mathcal{M}_{\mathcal{S}_h}$, $\bar{A}_h$ and $\bar{A}_h^T$ essentially represent the relation matrices respectively on the left and right side of the symmetry axis of $\mathcal{S}_h$.
When using notes \ref{locality} and \ref{approximation}, we must explicitly distinguish them.
Before proceeding, let $C_{L'_0}$ denote the set of objects belonging to the target object type, and $W_{L_{h-1}L'_0}$ denote the relation matrix from the updated $C_{L_{h-1}}$ to $C_{L'_0}$.
The left-hand relation matrix can be denoted as $\bar{A}_h^l=\prod_{i=0}^{\frac{h}{2}-1}W_{L_iL_{i+1}}$,
and the right-hand relation matrix can be denoted as $\bar{A}_h^r=\left(\prod_{i=\frac{h}{2}}^{h-2}W_{L_iL_{i+1}}\right)W_{L_{h-1}L'_0}$.
As a result,
\begin{equation}\label{modified_normalized_commuting_matrix_dms}
\begin{array}{lll}
\mathcal{\bar{M}}_{\mathcal{D}}^{new} & = & \sum_{h=2}^{2h_0-2}w_{\frac{h}{2}-1}\bar{A}_h^l\bar{A}_h^r \\
 &  & +w_{h_0-1}(1-\lambda)\bar{A}_{2h_0}^l\left(\mathbf{I}-\lambda\bar{R}\right)^{-1}\bar{A}_{2h_0}^r,
\end{array}
\end{equation}
where
$\bar{R}=W_{L_{h_0}L_{h_0+1}}W_{L_{h_0}L_{h_0+1}}^T$ still denotes the relation matrix of the recurrent substructure.
$SMSS$ is still defined via formula \ref{dms_sim} using $\mathcal{\bar{M}}_{\mathcal{D}}^{new}$.

\begin{note}\label{substitution}
In practice, it is very time-consuming to compute $(\mathbf{I}-\lambda\bar{R})^{-1}$ in formula \ref{modified_normalized_commuting_matrix_dms}.
Note that $\bar{A}_{2h_0}^l$ is a row vector according to the locality strategy.
Therefore, computing $\bar{A}_{2h_0}^l\left(\mathbf{I}-\lambda\bar{R}\right)^{-1}$ is equivalent to solving the linear equations $(\mathbf{I}-\lambda\bar{R})^TX=\bar{A}_{2h_0}^{lT}$.
We apply Lower-Upper (\textbf{LU}) decomposition to $\mathbf{I}-\lambda\bar{R}$ and then use the \textbf{LU} factors to obtain $X$ \cite{numerical_analysis}.
\end{note}

\subsection{Algorithm Description}\label{subsec:algorithmdescription}

Now, we describe the algorithm for computing $SMSS$, see algorithm \ref{DMSSALG}. As shown in lines 2-5, it takes at most $O_1=O(h_0\sum_{h=1}^{h_0-1}|C_{L_{h-1}}||C_{L_h}||C_{L_{h+1}}|+|C_{L_{h_0}}|^2|C_{L_{h_0+1}}|)$
to construct $\bar{A}_h^l$ and $\bar{A}_h^r$ for $h=2,4,\cdots,2h_0$ and compute $\mathcal{\bar{M}}_{\mathcal{D}}^{new}$.
According to note \ref{locality}, we obtain a vector whose entries represent the similarities between the source object and the others.
As shown in lines 6-9, it takes at most $O_2=|\phi(o_s)|^2+|\phi(o_s)||\phi(o_t)|$ to compute $SMSS(o_s,o_t)$. As a result, the worst-case time complexity of algorithm \ref{DMSSALG} is equal to
$O_1+O_2$.

\begin{algorithm}\footnotesize
\caption{Computing $SMSS$}
\begin{algorithmic}[1]
\label{DMSSALG}
\renewcommand{\algorithmicrequire}{\textsc{Input:}}
\renewcommand{\algorithmicensure}{\textsc{Output:}}
\REQUIRE{HIN $G$, source object $o_s$, decaying parameter $\lambda$, weights $w_h, 0,1,\cdots,h_0-1$.}
\ENSURE{Similarity vector $Sim$}
\STATE Compute $h_0$
       \FOR{$h=2,4,\cdots,2h_0$}
\STATE \quad Compute $\bar{A}_h^l$ and $\bar{A}_h^r$
       \ENDFOR
\STATE Compute $\mathcal{\bar{M}}_{\mathcal{D}}^{new}$ using formula \ref{modified_normalized_commuting_matrix_dms}
       \FOR{$o_t\in\phi(o_s)$}
\STATE Compute $SMSS(o_s,o_t)$ using formula \ref{dms_sim}
\STATE Append $SMSS(o_s,o_t)$ to vector $Sim$
       \ENDFOR
       \RETURN $Sim$
\end{algorithmic}
\end{algorithm}

Take the toy example shown in Fig. \ref{ExampleHins} as an example, its \textbf{SMS} is shown in Fig. \ref{ExampleDeepMetaStructure}(a).
$A$ is the source object type. Obviously, $h_0=2$, and $L_0=\{A\_0\}$, $L_1=\{P\_1\}$, $L_2=\{V\_2, T\_2\}$.
Therefore, $\mathcal{M}_{\mathcal{S}_2}=A_2A_2^T$, $\mathcal{M}_{\mathcal{S}_4}=A_4A_4^T$,
and $\mathcal{M}_{\mathcal{S}_6}=A_4RA_4^T$ and so on. Let $\lambda=0.999$ $w_0=0.3, w_1=0.7$.
$\mathcal{\bar{M}}_{\mathcal{D}}$ can be easily computed according to formula \ref{normalized_commuting_matrix_dms}, i.e.
\begin{displaymath}
\mathcal{\bar{M}}_{\mathcal{D}}=
\begin{bmatrix}
0.281 & 0.141 & 0.0 & 0.0 & 0.0 \\
0.141 & 0.125 & 0.073 & 0.073 & 0.110 \\
0.0 & 0.073 & 0.146 & 0.146 & 0.145 \\
0.0 & 0.073 & 0.146 & 0.146 & 0.145 \\
0.0 & 0.110 & 0.146 & 0.146 & 0220
\end{bmatrix}
\end{displaymath}
As a result,
\begin{displaymath}
SMSS=
\begin{bmatrix}
1.0 & 0.347 & 0.0 & 0.0 & 0.0 \\
0.347 & 1.0 & 0.269 & 0.269 & 0.319 \\
0.0 & 0.269 & 1.0 & 1.0 & 0.396 \\
0.0 & 0.269 & 1.0 & 1.0 & 0.396 \\
0.0 & 0.319 & 0.396 & 0.396 & 1.0
\end{bmatrix}
.
\end{displaymath}
In the similarity matrix $SMSS$, the rows (or columns) respectively represent A:Chuan Shi, A:Philip S. Yu, A:Yizhou Sun, A:Jiawei Han, A:Xifeng Yan.
For the deep meta structure in Fig. \ref{ExampleDeepMetaStructure}(b), $Gene$ is selected as the source object type.
We only need to note that $L_1=\{T\_1, GO\_1, CC\_1\}$, $L_1'=\{CC\_1\}$ because the degrees of $T$ and $GO$ in the network schema are equal to 1.

\section{Experimental Evaluations}\label{sec:expriment}

In this section, we compare the proposed metric $SMSS$ with the state-of-the-art metrics in terms of clustering task and ranking task on two real datasets.
The configuration of my PC is Intel(R) Core(TM) i5-4570 CPU @ 3.20GHz and RAM 12GB.
The evaluation criterion for ranking is Normalized Discounted Cumulative Gain ($nDCG\in\left[0,1\right]$, the better, the larger) \cite{nDCG:2013},
for clustering is Normalized Mutual Information ($NMI\in\left[0,1\right]$, the better, the larger) \cite{SHYYW:2011}.

\subsection{Datasets}\label{subsec:dataset}

Two real datasets, called \textbf{DBLPr} and \textbf{BioIN}, are used here.
The first is extracted from \textbf{DBLP}\footnote{http://dblp.uni-trier.de/db/}, and another is extracted from \textbf{Chem2Bio2RDF} \cite{CDJWZ:2010,FDSCSB:2016}.
\textbf{DBLPr} includes 30 venues coming from six areas: database, data mining, information retrieval, information system, web mining and web information management,
and 26549 papers, 26081 authors, 16798 terms. Its network schema is shown in Fig. \ref{exampleNetworkSchema}(a).
\textbf{BioIN} includes 2018 genes, 300 tissues, 4331 gene ontology instances, 224 substructures, 712 side effects and 18097 chemical compounds.
Its network schema is shown in Fig. \ref{exampleNetworkSchema}(b).
Note in particular that genes assigned to multiple clusters are not considered here.
That means each gene in BioIN is assigned to a single cluster.
The \textbf{SMS} for \textbf{DBLPr} and \textbf{BioIN} are respectively shown in Fig. \ref{ExampleDeepMetaStructure}(a,b).

\subsection{Baselines}\label{subsec:baseline}

In this paper, $SMSS$ is compared with three state-of-the-art similarity metrics: $BSCSE$ \cite{HZCSML:2016}, $BPCRW$ \cite{LC:2010b,LC:2010a}, $PathSim$ \cite{SHYYW:2011}.
Let $\mathcal{P}$ and $\mathcal{S}$ respectively denote a meta path and a meta structure. For a given source-target object pair $(o_s,o_t)$, they are defined as follows.
\begin{enumerate}
  \item $BSCSE(g,i|\mathcal{S},o_t)=\frac{\sum_{g'\in\sigma(g,i|\mathcal{S},G)}BSCSE(g',i+1|\mathcal{S},o_t)}{|\sigma(g,i|\mathcal{S},G)|^{\lambda}}$;
  \item $BPCRW(o,o_t|\mathcal{P})=\frac{\sum_{o'\in N_{\mathcal{P}}(o)}BPCRW(o',o_t|\mathcal{P})}{|N_{\mathcal{P}}(o)|^{\lambda}}$;
  \item $PathSim(o_s,o_t|\mathcal{P})=\frac{2\times\mathcal{M}_{\mathcal{P}}(o_s,o_t)}{\mathcal{M}_{\mathcal{P}}(o_s,o_s)+\mathcal{M}_{\mathcal{P}}(o_t,o_t)}$.
\end{enumerate}
In these definitions, $\lambda$ is a biased parameter.
For $BSCSE$, $\sigma(g,i|\mathcal{S},G)$ denotes the ($i+1$)-th layer's instances expanded from $g\in\mathcal{S}[1:i]$ on $G$ \cite{HZCSML:2016}.
For $BPCRW$, $N_{\mathcal{P}}(o)$ denotes the neighbors of $o$ along meta path $\mathcal{P}$ \cite{LC:2010b,LC:2010a}.
For $PathSim$, $\mathcal{M}_{\mathcal{P}}$ denotes the commuting matrix of the meta path $\mathcal{P}$ \cite{SHYYW:2011}.

\subsection{Parameter Setting}\label{subsec:parameter_setting}

$SMSS$, $BSCSE$ and $BPCRW$ involve some input parameters.
For a specific task, we can first extract a subgraph from the complete \textbf{HIN} as a validation set,
and then select the optimal parameters for the task on the subgraph.
In this paper, we investigate the ranking quality and clustering quality under different input parameters of $SMSS$
and under ``optimal parameters'' of $BSCSE$ and $BPCRW$.
Specifically, $\lambda$ for $BSCSE$ and $BPCRW$ is respectively set to $0.1,0.3,0.5,0.7,0.9$.
The ``optimal $\lambda$''  means the one maximizing $nDCG$ or $NMI$ under these settings.
For $SMSS$, $\lambda$ is respectively set to $0.1,0.3,0.5,0.7,0.9$.
And $w_0,w_1$ for each $\lambda$ are sampled from the Beta distribution respectively with hyper-parameters $(1,9)$, $(2,8)$, $(3,7)$, $(4,6)$, $(5,5)$, $(6,4)$, $(7,3)$, $(8,2)$, $(9,1)$.
The optimal value ($NMI$ or $nDCG$) within these samples ($w_0$ and $w_1$) is selected as the one corresponding to this $\lambda$.
``optimal $(\lambda,w_0,w_1)$'' for $SMSS$ means the ones maximizing $nDCG$ or $NMI$ under all possible settings of $\lambda,w_0,w_1$.

\subsection{Clustering Quality}\label{subsec:clustering}

Now, we compare $SMSS$ with the baselines in terms of clustering quality ($NMI$ \cite{SHYYW:2011}, the higher, the better) on \textbf{DBLPr} and \textbf{BioIN}.
First, we compute the similarities between two objects respectively using these metrics.
That means we obtain a feature vector for each object. Then, $k$-means method is used to complete the clustering task using these feature vectors.
For \textbf{DBLPr}, $Venue$ is selected as the source and target object type. Its benchmark is given according to the scope of the venues.
For \textbf{BioIN}, $Gene$ is set to the source and target object type. Its benchmark is extracted from the one used in paper \cite{benchmark}.
$k$ is set to the number of clusters in the benchmark.

\subsubsection{On BioIN}\label{subsubsec:clustering_slap}

\begin{table}\scriptsize
  \centering
  \caption{Optimal $NMI$ on BioIN.}
  \begin{tabular}{c|c|c}
    \hline
     \textbf{Metric} & \textbf{Schematic Structure} & $\mathbf{NMI}$ \\
    \hline
     $SMSS$ (optimal $\lambda,w_0,w_1$) & - & $\mathbf{0.79514}$ \\
    \hline
    \multirow{2}{*}{$BSCSE$ (optimal $\lambda$)} & (G,(GO,T),G) & 0.74760 \\ \cline{2-3}
                            & (G,CC,(Si,Sub),CC,G) & 0.32577 \\ \cline{2-3}
    \hline
    \multirow{4}{*}{$BPCRW$ (optimal $\lambda$)}  & (G,T,G) & 0.36543 \\ \cline{2-3}
                              & (G,GO,G) & 0.78456 \\ \cline{2-3}
                              & (G,CC,Si,CC,G) & 0.22343 \\ \cline{2-3}
                              & (G,CC,Sub,CC,G) & 0.31246 \\ \cline{2-3}
    \hline
    \multirow{4}{*}{$PathSim$}  & (G,T,G) & 0.35910 \\ \cline{2-3}
                                & (G,GO,G) & 0.78847 \\ \cline{2-3}
                                & (G,CC,Si,CC,G) & 0.21397 \\ \cline{2-3}
                                & (G,CC,Sub,CC,G) & 0.30809 \\ \cline{2-3}
    \hline
  \end{tabular}
  \label{SlapClustering}
\end{table}

Table \ref{SlapClustering} shows the optimal $NMI$ values for $SMSS$ and the baselines with different schematic structures on \textbf{BioIN}.
Note in particular that $SMSS$ does not depend on any meta paths and meta structures. Therefore, its cell corresponding to schematic structure is filled with ``-''.
The optimal $\lambda$ for $BSCSE$ and $BPCRW$ and optimal $(\lambda,w_0,w_1)$ for $SMSS$ are settled using the method in section \ref{subsec:parameter_setting}.
Obviously, the optimal $NMI$ for $SMSS$ is larger than those for the baselines with different schematic structure, especially when
choosing the meta structure $(G,CC,(Si,Sub),CC,G)$ and the meta paths $(G,T,G)$, $(G,CC,Si,CC,G)$ and $(G,CC,Sub,CC,G)$.


\begin{figure}[htb]
  \centering
  \includegraphics[width=0.6\textwidth]{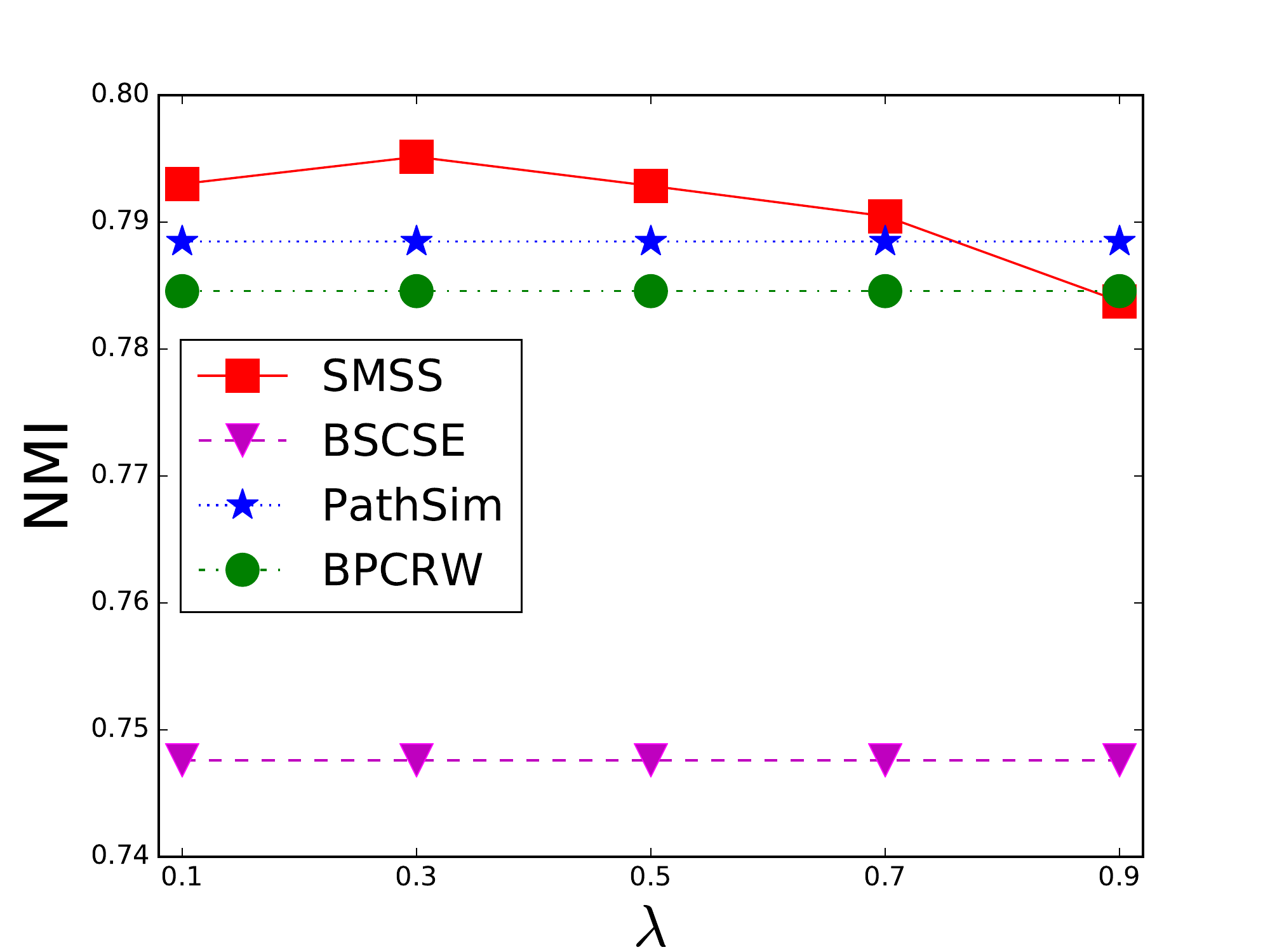}
  \caption{Comparison of $NMI$ for $SMSS$ under different $\lambda$ with optimal $NMI$ for $PathSim$, $BSCSE$ and $BPCRW$ on BioIN.}
  \label{SlapClusterLambda}
\end{figure}

Now, we examine whether $NMI$ values for $SMSS$ under different decaying parameter $\lambda$ are still larger than the optimal $NMI$ for $PathSim$, $BSCSE$ and $BPCRW$ on \textbf{BioIN}.
According to table \ref{SlapClustering}, $(G,GO,G)$ is selected as the meta paths for $PathSim$ and $BPCRW$, and $(G,(GO,T),G)$ is selected as the meta structure for $BSCSE$.
The decaying parameter $\lambda$ of $SMSS$ is respectively set to 0.1, 0.3, 0.5, 0.7, 0.9, and the corresponding values for each $\lambda$ can be settled using the method in section \ref{subsec:parameter_setting}.
Fig. \ref{SlapClusterLambda} presents the comparisons of $NMI$ for $SMSS$ under different $\lambda$ with optimal $NMI$ for $PathSim$, $BSCSE$ and $BPCRW$.
From this figure, we know that $NMI$ for $SMSS$ is always larger than that for $PathSim$, $BSCSE$ and $BPCRW$ when $\lambda=0.1,0.3,0.5,0.7$.
When $\lambda=0.9$, $NMI$ for $SMSS$ is a little larger than that for $BPCRW$, but a little less than that for $PathSim$.
According to table \ref{SlapClustering}, we know $(G,GO,G)$ plays a more important role than other meta paths (e.g. $(G,CC,Si,CC,G)$ and $(G,CC,Sub,CC,G)$) in \textbf{BioIN} in terms of the clustering task.
According to equation \ref{normalized_commuting_matrix_dms} and the \textbf{SMS} shown in Fig. \ref{ExampleDeepMetaStructure}(b),
the weight of $(G,GO,G)$ is smaller than the weights of the other schematic structures as $\lambda$ becomes large. Therefore, $NMI$ for $SMSS$ drops a little at this time.

\subsubsection{On DBLPr}\label{subsubsec:clustering_dblp}

\begin{table}\scriptsize
  \centering
  \caption{Optimal $NMI$ values on DBLPr.}
  \begin{tabular}{c|c|c}
    \hline
     \textbf{Metric} & \textbf{Schematic Structure} & $\mathbf{NMI}$ \\
    \hline
     $SMSS$ (optimal $\lambda,w_0,w_1$) & - & $\mathbf{0.88449}$ \\
    \hline
    $BSCSE$ (optimal $\alpha$) & (V,P,(A,T),P,V) & 0.76927 \\
    \hline
    \multirow{2}{*}{$BPCRW$ (optimal $\alpha$)}  & (V,P,A,P,V) & 0.80348 \\ \cline{2-3}
                              & (V,P,T,P,V) & 0.78964 \\ \cline{2-3}
    \hline
    \multirow{2}{*}{$PathSim$}  & (V,P,A,P,V) & 0.65622 \\ \cline{2-3}
                                & (V,P,T,P,V) & 0.66509 \\ \cline{2-3}
    \hline
  \end{tabular}
  \label{DblpClustering}
\end{table}

Table \ref{DblpClustering} shows the optimal $NMI$ values for $SMSS$ and the baselines with different schematic structures on \textbf{DBLPr}.
The cell of $SMSS$ corresponding to schematic structure is also filled with ``-''. the optimal $\lambda$ for $BSCSE$ and $BPCRW$ and optimal $(\lambda,w_0,w_1)$ for $SMSS$
can also be settled using the method in subsection \ref{subsec:parameter_setting}.
Obviously, the optimal $NMI$ for $SMSS$ is larger than those for the baselines with different schematic structures.


\begin{figure}[htb]
  \centering
  \includegraphics[width=0.6\textwidth]{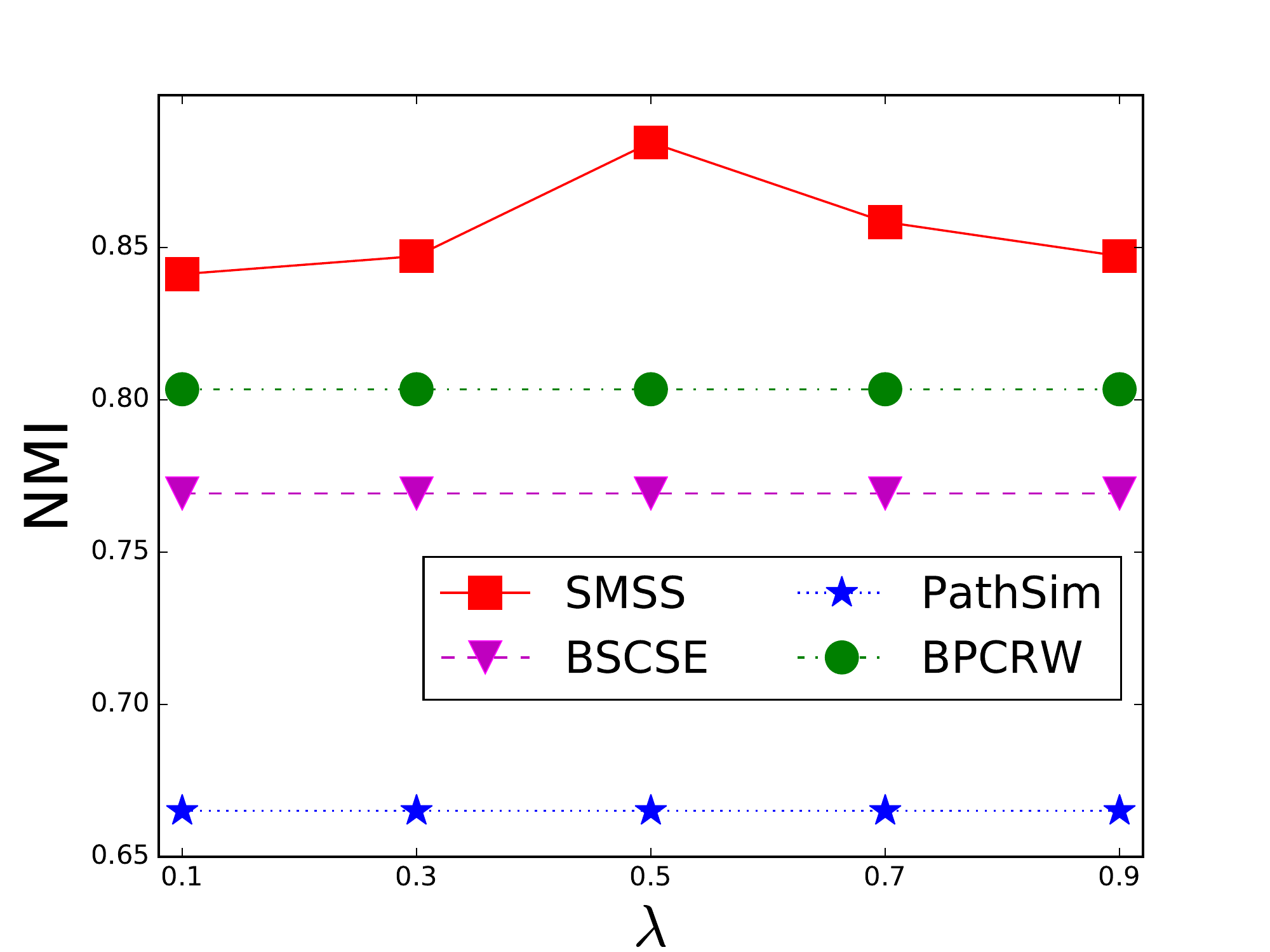}
  \caption{Comparison of $NMI$ for $SMSS$ under different $\lambda$ with optimal $NMI$ for $PathSim$, $BSCSE$ and $BPCRW$ on DBLPr.}
  \label{DblpClusterLambda}
\end{figure}

Now, we examine whether $NMI$ for $SMSS$ under different decaying parameter $\lambda$ is still larger than optimal $NMI$ for $PathSim$, $BSCSE$ and $BPCRW$ on \textbf{DBLPr}.
According to table \ref{DblpClustering}, the meta paths for $PathSim$ is $(V,P,T,P,V)$ and for $BPCRW$ is $(V,P,A,P,V)$, and the meta structure for $BSCSE$ is $(V,P,(A,T),P,V)$.
The decaying parameter $\lambda$ of $SMSS$ is respectively set to 0.1, 0.3, 0.5, 0.7, 0.9, and the corresponding values for each $\lambda$ can be settled using the method in section \ref{subsec:parameter_setting}.
Fig. \ref{DblpClusterLambda} presents the comparisons of $NMI$ for $SMSS$ under different $\lambda$ with the optimal $NMI$ for $PathSim$, $BSCSE$ and $BPCRW$.
According to this figure, we know that $NMI$ for $SMSS$ is much larger than those for $PathSim$, $BSCSE$ and $BPCRW$ whenever $\lambda=0.1,0.3,0.5,0.7,0.9$.
Integrating Fig. \ref{SlapClusterLambda} and Fig. \ref{DblpClusterLambda}, we know that $SMSS$ significantly outperforms the baselines
especially when choosing different meta paths or meta structures.

\subsection{Ranking Quality}\label{subsec:ranking}

\begin{table*}
  \centering
  \caption{Optimal $nDCG$ values on \textbf{DBLPr}. (V,P,(A,T),P,V) denotes the meta structure shown in Fig. \ref{exampleMetaPath}(d).}
  \begin{tabular}{c|c|c|c|c}
    \hline
     \multicolumn{2}{c|}{Schematic Structure} & VLDB & ICDE & SIGMOD \\
    \hline
    \multirow{3}{*}{$(V,P,A,P,V)$} & $BPCRW$ & \multirow{2}{*}{0.96354} & \multirow{2}{*}{0.95975} & \multirow{2}{*}{0.96554} \\
      & (optimal $\alpha$) &  &  &  \\ \cline{2-5}
      & $PathSim$ & 0.97831 & 0.97231 & 0.98069 \\
    \hline
    \multirow{3}{*}{$(V,P,T,P,V)$} & $BPCRW$ & \multirow{2}{*}{0.91979} & \multirow{2}{*}{0.91523} & \multirow{2}{*}{0.93000} \\
      & (optimal $\alpha$) & & &  \\ \cline{2-5}
      & $PathSim$ & 0.94248 & 0.94871 & 0.94993 \\
    \hline
    \multirow{2}{*}{$(V,P(A,T),P,V)$} & $BSCSE$ & \multirow{2}{*}{0.96868} & \multirow{2}{*}{0.96419} & \multirow{2}{*}{0.97147} \\
      & (optimal $\alpha$) &  &  &  \\
    \hline
    \multirow{2}{*}{-} & $SMSS$ & \multirow{2}{*}{$\mathbf{0.98185}$} & \multirow{2}{*}{$\mathbf{0.97254}$} & \multirow{2}{*}{$\mathbf{0.98111}$} \\
      & (optimal $\lambda,w_0,w_1$) &  &  & \\
    \hline
  \end{tabular}
  \label{RankingEval}
\end{table*}


Here, we compare $SMSS$ with the baselines in terms of ranking quality ($nDCG$ \cite{SHYYW:2011}, the higher, the better) on \textbf{DBLPr}.
$Venue$ is selected as the source object type, and ICDE, VLDB and SIGMOD are selected as source objects. We respectively compute the similarities from the source objects to other objects.
For each source, all the objects can be ranked 0 (unrelated), 1 (slightly related), 2 (fairly related), 3 (highly related) according to the relatedness to this source.
Thus, we can compute $nDCG$ respectively for each source object. The optimal $\lambda$ for $BSCSE$ and $BPCRW$ and optimal $(\lambda,w_0,w_1)$ for $SMSS$ are settled using the method in
subsection \ref{subsec:parameter_setting}.

Table \ref{RankingEval} shows the optimal $nDCG$ values for $SMSS$ and the baselines with different meta paths or meta structures.
For ICDE, VLDB, SIGMOD, the optimal $nDCG$ for $SMSS$ is always larger than those for the baselines with different meta paths or meta structures.

\begin{figure}[htb]
  \centering
  \subfigure[ICDE]{%
  \includegraphics[width=0.325\textwidth]{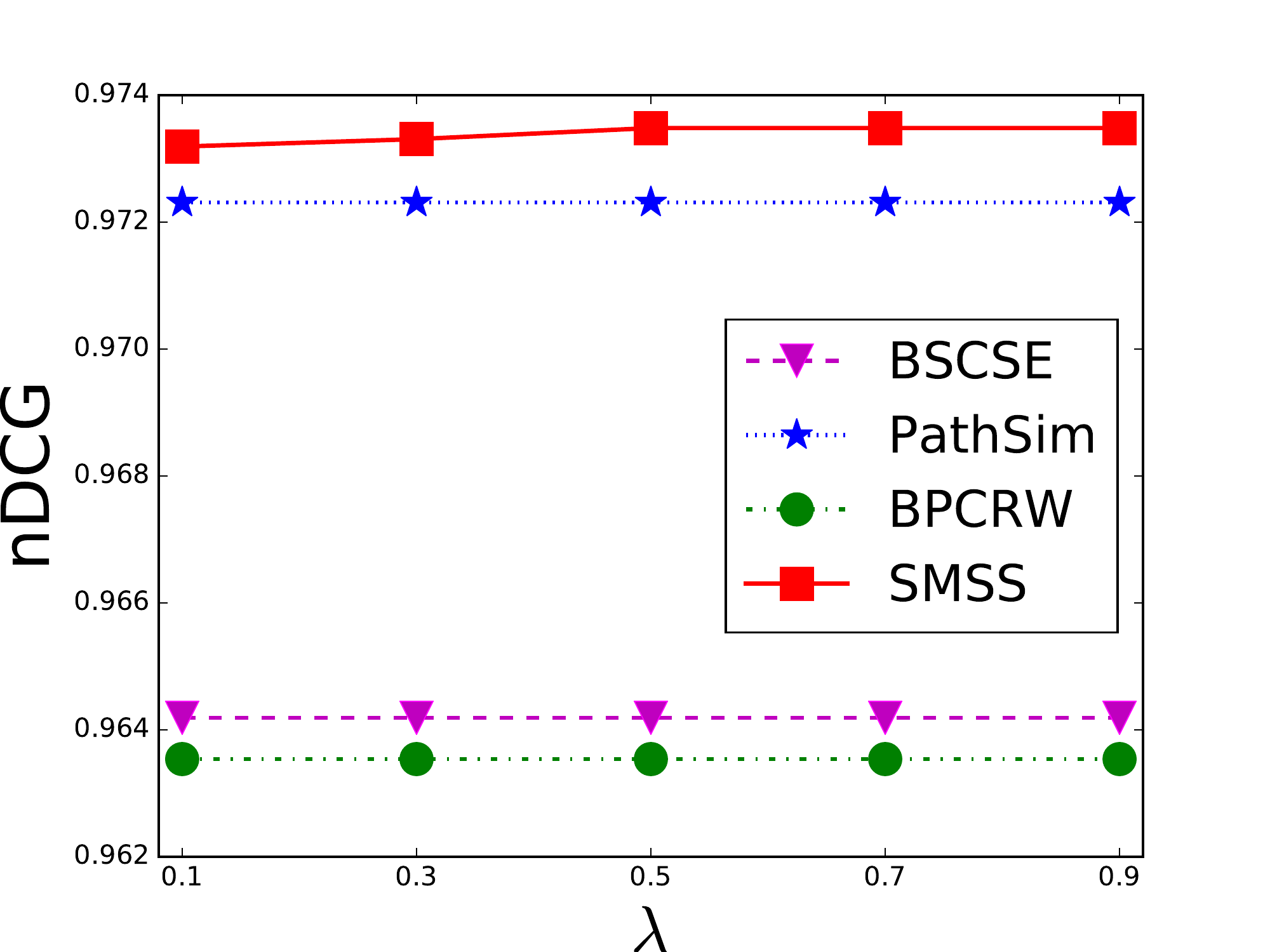}}
  \subfigure[VLDB]{%
  \includegraphics[width=0.325\textwidth]{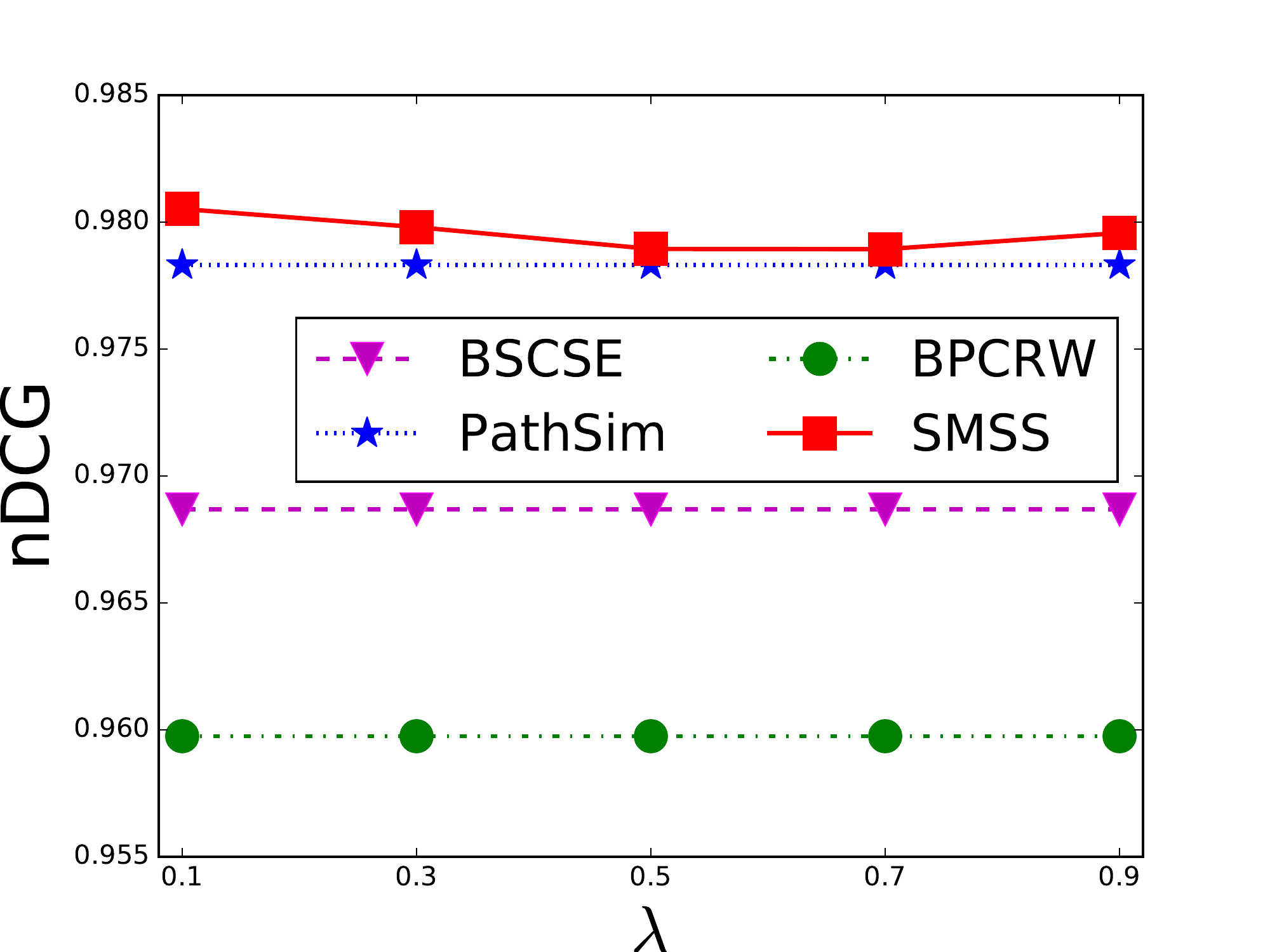}}
  \subfigure[SIGMOD]{%
  \includegraphics[width=0.325\textwidth]{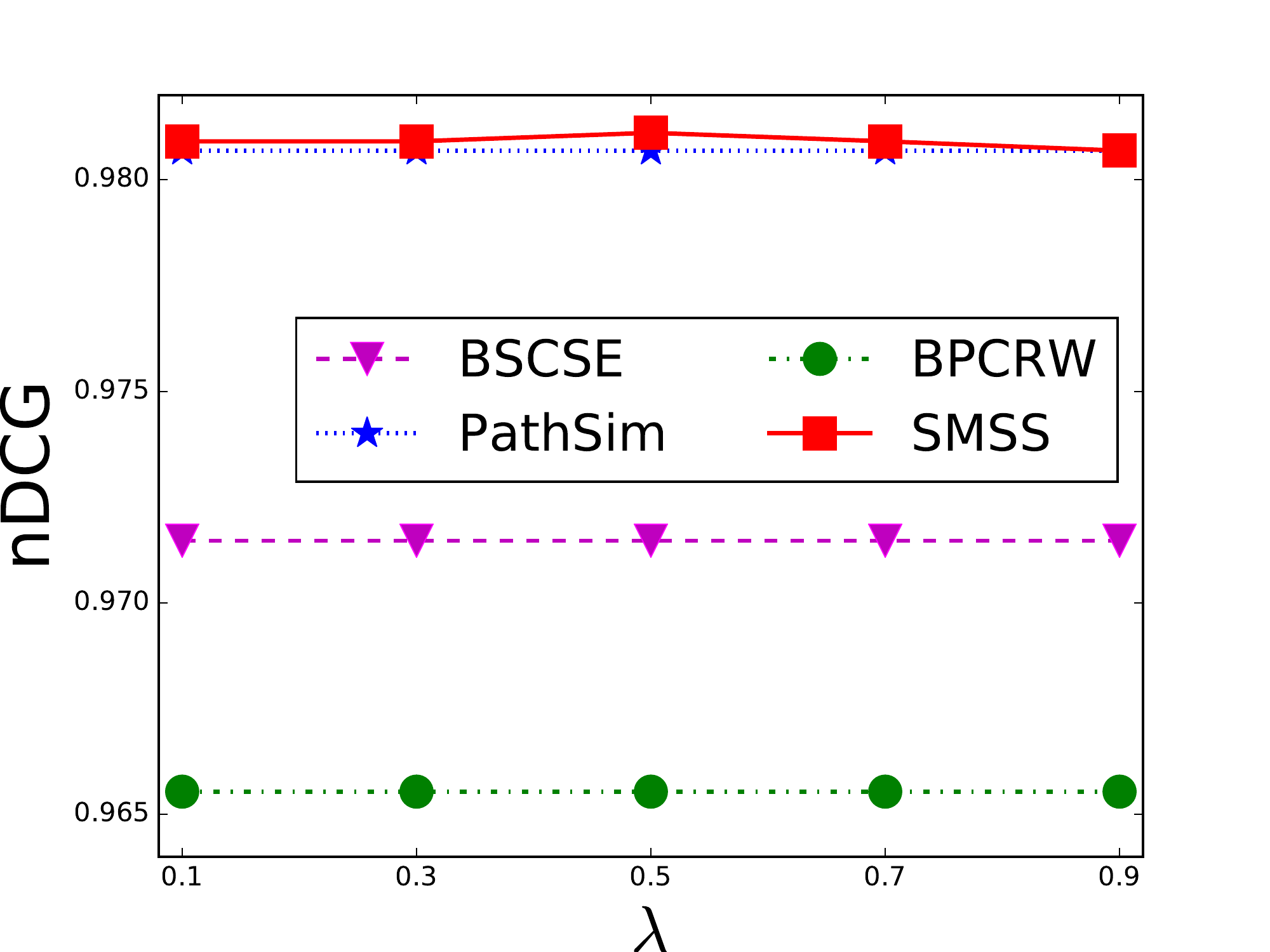}}
  \caption{Comparison of $nDCG$ for $SMSS$ under different $\lambda$ with optimal $nDCG$ for $BPCRW$, $PathSim$ and $BSCSE$.}
  \label{RankingLambda}
\end{figure}

Now, we examine whether $nDCG$ for $SMSS$ is still larger than the optimal $nDCG$ for the baselines under different $\lambda$ on \textbf{DBLPr}.
According to table \ref{RankingEval}, $(V,P,A,P,V)$ is selected as the meta path for $PathSim$ and $BPCRW$, and $(V,P,(A,T),P,V)$ as the meta structure for $BSCSE$.
The decaying parameter $\lambda$ of $SMSS$ is respectively set to 0.1, 0.3, 0.5, 0.7, 0.9, and the corresponding values for each $\lambda$ can be settled using the method in section \ref{subsec:parameter_setting}.
Fig. \ref{RankingLambda} shows the comparisons of $nDCG$ for $SMSS$ under different decaying parameters $\lambda$ with optimal $nDCG$ for $PathSim$, $BSCSE$ and $BPCRW$.
For ICDE and VLDB, their $nDCG$ values for $SMSS$ are much larger than those for $PathSim$, $BSCSE$ and $BPCRW$ when $\lambda=0.1,0.3,0.5,0.7,0.9$.
For SIGMOD, its $nDCG$ values for $SMSS$ are a little larger than those for $PathSim$, $BSCSE$ and $BPCRW$ when taking different $\lambda$.
In conclusion, the proposed metric $SMSS$ outperforms the baselines in terms of ranking on the whole.

\subsection{Time Efficiency}\label{subsec:timeefficiency}

\begin{table}
  \centering
  \caption{Average Running Time (Sec) of computing similarity between a source and a target object.}
  \begin{tabular}{c|c|c}
    \hline
     \diagbox{\textbf{Metric}}{\textbf{DataSet}} & \textbf{DBLPr} & \textbf{BioIN} \\
     \hline
      $BPCRW$ & $\mathbf{0.27005}$ & 0.01634 \\
     \hline
      $PathSim$ & 5.39970 & 0.94207 \\
     \hline
      $SMSS$ & 68.29283 & 0.97953 \\
     \hline
      $BSCSE$ & 610.68539 & $\mathbf{0.00395}$ \\
    \hline
  \end{tabular}
  \label{RunningTime}
\end{table}

Here, we evaluate the time efficiency of $SMSS$, $BPCRW$, $PathSim$ and $BSCSE$ on \textbf{BioIN} and \textbf{DBLPr}.
Table \ref{RunningTime} shows the running time of computing the similarities between a source object and a target object using $SMSS$ and the baselines on \textbf{BioIN} and \textbf{DBLPr}.
According to this table, $SMSS$ has no advantages in terms of running time. This is consistent with our expectation, because computing $SMSS$ requires a lot of matrix operations.
As stated previously, $SMSS$ does not depend on any meta paths and meta structures.
This advantage is obtained at the sacrifice of computational efficiency.
In fact, we can employ Graphics Processing Unit (GPU) to accelerate the matrix operations.
However, we have no GPU. So we cannot implement our algorithm using GPU.

\section{Conclusion}\label{sec:conclusion}

In this paper, we propose a stratified meta structure based similarity $SMSS$ in \textbf{HINs}.
The stratified meta structure can be automatically constructed by repetitively traversing the network schema, and contains rich semantics.
That means users do not worry about selecting an inappropriate meta paths or meta structures.
To formalize the semantics in the \textbf{SMS}, we firstly use the commuting matrices to formalize the semantics in meta structures,
and then define the commuting matrix of the \textbf{SMS} by combining all the commuting matrices of the relevant meta structures or meta paths.
Experimental evaluations show that $SMSS$ on the whole outperforms the baselines in terms of clustering and ranking.

\section*{References}
\bibliography{mybibfile}

\end{document}